\documentclass[11pt]{article}
\usepackage{graphicx}
\usepackage{lscape}
\usepackage{amsmath}
\usepackage{amsthm}
\usepackage{graphicx}
\usepackage{amsfonts}
\usepackage{amssymb}

\usepackage{natbib}

\newcommand{\EE}{\mathbb{E} }

\newcommand{\PP}{\mathbb{P} }

\usepackage{color}

\newcommand{\be}{\begin{equation}}
\newcommand{\en}{\end{equation}}

\newtheorem{theorem}{Theorem}
\newtheorem{prop}{Proposition}

\newtheorem{remark}{Remark}

\newcommand{\ea}{\end{eqnarray}}
\newcommand{\ba}{\begin{eqnarray}}
\newcommand{\ean}{\end{eqnarray*}}
\newcommand{\ban}{\begin{eqnarray*}}

\usepackage[margin=1.25in]{geometry}

\begin{document}

\title{SYSTEMIC RISK AND INTERBANK LENDING}

\author{Li-Hsien Sun\thanks{Institute of Statistics, National Central University, Chung-Li, Taiwan, 32001 {\em lihsiensun@ncu.edu.tw}. Work supported by Most grant 105-2118-M-008-005}}
 
\date{\today}
\maketitle

\bigskip

\begin{abstract}
%We propose a simple model of the banking system incorporating a game feature where the evolution of monetary reserve is  modeled as a system of coupled Feller diffusions. The optimization subject to the quadratic cost  reflects the desire of each bank to borrow from or lend to a central bank through manipulating its lending preference and  the intention of each bank to deposit in the central bank in order to control the volatility for cost minimization. We observe that the Markov Nash equilibrium generated by Cox-Ingersoll-Ross type processes creates  liquidity and deposit rate. The adding liquidity leads to a flocking effect implying stability or systemic risk depending on the level of the growth rate but the deposit rate diminishes the growth of the total monetary reserve causing a large number of bank defaults.
%The central bank acts as a central deposit corporation. In addition, 
%the corresponding Mean Field Game in the case of the number of banks $N$ large and 
%the infinite time horizon stochastic game with the discount factor are also discussed. 

%Finally, we solve for the closed-loop equilibria in the case of inter-bank lending and borrowing with clearing debt obligations using the stochastic game with delay. 

We propose a simple model of the banking system incorporating a game feature where the evolution of monetary reserve is modeled as a system of coupled Feller diffusions. The Markov Nash equilibrium generated through minimizing the linear quadratic cost subject to Cox-Ingersoll-Ross type processes creates liquidity and deposit rate. The adding liquidity leads to a flocking effect but the deposit rate diminishes the growth rate of the total monetary reserve causing a large number of bank defaults. In addition, the corresponding Mean Field Game and the infinite time horizon stochastic game with the discount factor are also discussed. 

\end{abstract}


\section{Introduction}\label{Intro}
The recent financial crisis motivates the research on systemic risk. Central to this research is constructing a safe banking system that alerts investors to possible systemic risk and tackle a financial crisis should it occur. However, due to the complexity of connection between banks, systemic risk is often unpredictable and unmeasurable. Hence, in order to better understand the banking system and systemic risk, we propose a  model of interbank lending and borrowing with a game feature where the evolution of monetary reserve is driven by a system of interacting Feller diffusions  and the default is simply defined as the zero monetary reserve as studied by \cite{Fouque-Ichiba}. Each bank controls its own rate of lending to or borrowing from the central bank to minimize its cost or maximize its profit. The spirit of optimal strategy with respect to the object function is that each bank intends to borrow when its monetary value falls below a critical level (the averaged capitalization)  and to lend when its monetary value remains above the critical level. Intuitively, banks intend to maximize their borrowing or lending activities during a fixed period. Hence,  referring to \cite{R.Carmona2013}, we restrict  the monetary amount that banks borrow from or lend to a central bank at the fixed incentive rate determined by individual banks or a regulator. To this end, we apply a quadratic cost function. The exact Markov Nash equilibrium for a linear quadratic regulator in the both finite and infinite time horizon and subject to Cox-Ingersoll-Ross (CIR) type processes is obtained by a system of Hamilton-Jacobi-Bellman equations. In addition, we solve for $\epsilon$-Nash equilibrium in the mean field limit using the coupled backward Hamilton-Jacobi-Bellman (HJB) equation and the forward Kolmogorov equation for the value function and distribution of dynamics. Given the feature of CIR type processes, the admissibility of the equilibrium is worth to be discussed. From a financial perspective, compared to case of the Ornstein-Uhlenbeck (OU) type processes proposed in \cite{R.Carmona2013}, due to the volatility driven by the squared root function of states, the 
Markov Nash equilibrium generated by CIR type processes creates not only liquidity but deposit rate impairing the growth rate leading to the system instability such that a regulator must govern the growth rate by controlling the corresponding lending/borrowing incentive to prevent system crash.  

Such interactions leading to multiple defaults is investigated by other mathematical models in continuous time.    \cite{Fouque-Sun, R.Carmona2013} comment on systemic risk using the interacting OU type processes. 
%\cite{Carmona-Fouque2016} consider the corresponding obligation of the lending and borrowing through delayed controls in the same type of processes. 
\cite{Garnier-Mean-Field, GarnierPapanicolaouYang} study the rare events in the potential bistable model based on large deviation results.
%In particular, the stability created by a central agent according to the model in \cite{GarnierPapanicolaouYang, Garnier-Mean-Field} is provided by \cite{Papanicolaou-Yang2015}. 
Mean Field Games (MFGs) proposed by \cite{MFG1, MFG2, MFG3} are applied frequently to obtain the asymptotic equilibria called $\epsilon$-Nash equilibria of stochastic differential games with interactions given by empirical distributions whose solution satisfies the strongly coupled  HJB equations backward in time and the Kolmogorov equations forward in time. Independently, Caines, Huang, and Malham\'{e} develop a similar scenario called the Nash Certainty equivalence. See \cite{HuangCainesMalhame1,HuangCainesMalhame2} for instance. Instead of the partial differential equation (PDE) approach, \cite{Bensoussan_et_al, CarmonaDelarueLachapelle, Carmona-Lacker2015} consider the probabilistic approach using appropriate forms of adjoint forward-backward stochastic differential equations (FBSDEs) to solve $\epsilon$-Nash equilibria satisfying the Pontryagin stochastic minimum principle. In particular, \cite{Lacker-Webster2014, Lacker2015} study MFGs with common noise. Due to the volatility in Feller diffusions driven by the square-root function of states, the sufficient Pontryagin minimum principle is not straightforward to be applied. Hence, we tackle the optimization problem using the HJB approach. 

Consider $N$ banks that lend to and borrow from each other. Their monetary reserves are represented by a system of diffusion processes $X^i_t$ for $i=1,\cdots,N$ driven by $N$ independent standard Brownian motions $W^i_t$, $i=1,\cdots,N$. The initial values of the system may be nonnegative squared integrable and identically distributed random variables $X^i_0=\xi^i$ for $i=1,\cdots,N$ independent of the Brownian motions. The model is written as
\be
d X^{i}_t = \left(\frac{a}{N}\sum_{j=1}^N(  X^j_t- X^i_t)+\gamma_t\right)dt+\alpha^i_t dt
 +2\sqrt{X^i_t}d  {W}^{i}_t\,, \quad  i=1,\cdots, N,
\en
where the nonnegative growth rate or income rate  $\gamma_t $ for the individual bank is a deterministic function in $L^{\infty}$ space and the lending preference, or rate of lending and borrowing is normalized by the number of banks with $0\leq a\leq N$. Referring to \cite{Fouque-Ichiba}, we further assume the diffusion parameter equal to $2$ identical to squared Bessel processes in order to simplify the analysis of system stability. 
%It is easy to extend to the case with general parameter $\sigma$.   

The drift terms represent interactions between the banks; specifically they define the borrowing or lending rate of bank $i$ from or to bank $j$. The simple lending and borrowing can be rewritten in the mean field form as  
\be
dX^i_t=\left(a(\overline X_t-X^i_t)+\gamma_t+\alpha^i_t\right) dt +2\sqrt{X^i_t}d  {W}^{i}_t\,, \quad  i=1,\cdots, N,
\en
where  $\overline X_t=\frac{1}{N}\sum_{j=1}^NX^j_t$ represents the averaged capitalization whereby bank $i\in\{1,\cdots,N\}$ chooses its own optimization strategy $\alpha^i$ for optimizing its lending and borrowing rates to/from the central bank at initial time by minimizing the quadratic cost   
\ba\label{objective-intro}
J^i(\alpha^1,\cdots,\alpha^N)=\EE\bigg\{\int_{0}^{T} f_i(X_t,\alpha^i_t)
dt+g_i(X_{T}^{i})\bigg\},
\ea
where $X=(X^1,\cdots,X^N)$, $x=(x^1,\cdots,x^N)$, $\alpha=(\alpha^1,\cdots,\alpha^N)$. Note that $\alpha_\cdot$ is a progressively measurable control process and  $\alpha^i_\cdot$ is admissible if for $0\leq t\leq T$,
\be\label{admissible-cond}
X^{j,\alpha^i}_t\geq 0,\; \forall j\neq i\; {\mathrm{and}}\;X^{i,\alpha^i}_t\geq 0\;a.s..
\en  
The running cost and the terminal cost are defined as  
\[
f_i(X_t,\alpha^i_t)
=\frac{(\alpha^i_t)^2}{2}-q\alpha^i_t(\overline{X}_t-X^i_t)
+\frac{\epsilon}{2}(\overline{X}_t-X^i_t)^2,
%+\frac{(\alpha^i_{t-\delta})^2}{2}
%+\frac{\tilde\epsilon}{2}(\overline{X}_{t-\delta}-X^i_{t-\delta})^2
\]
and  
\[
g_i(X_T)=\frac{c}{2}\left(\overline{X}_T-X^{i}_T\right)^2,
\]
where $q$, $c$, and $\epsilon$ are positive parameters and $q^2\leq\epsilon$ in order to satisfy the convexity condition of the running cost $f_i(x,\alpha)$. Note that $q$ controls the borrowing and lending incentive and the case of $X^i_t\leq \overline X_t$ gives $\alpha^i_t\geq 0$ denoting that bank $i$ intends to borrow. Conversely, when $X^i_t > \overline X_t$, the negative strategy $\alpha^i_t<0$ implies a lending intention by bank $i$.  In addition, $\epsilon$ and $c$ are treated as the parameters for the running penalty and the terminal penalty respectively implying that banks intends to hold their own reserves close to the average capitalization. 

%In particular, $q^2\leq\epsilon$ leads to the convexity condition of the running cost $f_i(x,\alpha)$. 

The remainder of this paper is organized as follows. Section \ref{Interbank} discusses the simple uncontrolled constant growth rate of interbank lending and borrowing. The behavior of the system is studied using the results of \cite{Fouque-Ichiba}. Section \ref{Exact-Nash} is devoted to solving the Markov Nash equilibria of the stochastic differential games and the $\epsilon$-Nash equilibria in the mean field limit. 
We implement the coupled Hamilton-Jacobi-Bellman equations in Markovian models. The financial implication is illustrated at the end of this section. The infinite time horizon with the discount factor and its influence on the system crash and multiple defaults are analyzed in Section \ref{Inf-Nash}. The conclusion is provided in Section \ref{conclusion} and the verification theorems for both finite and infinite time horizons are given in the appendix.    

%In Section \ref{Delay}, we consider the model of lending and borrowing with the corresponding clearing debt obligations and search for the closed-loop equilibrium using  the dynamic programming-caratheodory approach. 

\section{Interbank Lending System}\label{Interbank}
In this section, we discuss the simple model of lending and borrowing without controls. In this model, bank $i$ has a monetary reserve $X^i_t$ with a constant growth rate $\gamma$ written as 
\be\label{coupled-no-game}
d X^{i}_t = \left(\frac{a}{N}\sum_{j=1}^N(  X^j_t- X^i_t)+\gamma\right)dt+2\sqrt{X^i_t}d  {W}^{i}_t\,, \quad  i=1,\cdots, N,
\en
 where the lending preference is a fixed normalized constant $\frac{a}{N}\leq 1$ treated as one particular case in \cite{Fouque-Ichiba}. Consequently, the total monetary reserve $Y_t=\sum_{i=1}^NX^i_t$ is given by
\be\label{Y-no-game}
dY_t=N\gamma dt+2\sqrt{  Y_t}d\widetilde W_t
\en
where $\widetilde W_t$ is a standard Brownian motion in some extension probability space by L\'evy theorem referring to Theorem 3.4.2 of \cite{Karatzas2000}. We straightforwardly apply the results of \cite{Fouque-Ichiba} (See also Chapter 6 of \cite{Yor2013} for instance) as follows:
\begin{itemize}
\item System crash
\begin{enumerate}
\item If $\gamma>\frac{2}{N}$, systemic risk never happens owing to the total monetary reserve $Y$ never reaching zero:
\be\label{cond-1}
\PP(Y_t>0\;\mathrm{for\;all\;} t\in[0,\infty))=1.
\en
\item If $\gamma=\frac{2}{N}$, we obtain 
\be\label{cond-2}
\PP(\limsup_{t\rightarrow\infty}Y_t=\infty)=1,
\en
and the system survivals since the total reserve $X$ remains positive . However, due to
\be\label{cond-3}
\PP(\inf_{0\leq t\leq\infty}Y_t=0)=1,
\en 
the shrinking of the total reserve to almost nothing  leads to a financial crisis at some point in the future almost surely. 
\item If $0<\gamma<\frac{2}{N}$, the total monetary reserve $Y$ violates the property \eqref{cond-2} but satisfies property \eqref{cond-3}. Under this condition, the total monetary reserve $Y$ reaches zero almost surely at some large time and reflects instantaneously at that point. Therefore, all banks face defaults in the future
\item If $\gamma=0$,  the total monetary reserve $Y$ approaches to zero within some finite time. At this time, all banks default at some time and thereafter remain at zero almost surely.
\end{enumerate}
%\item If the constant normalized rate $\frac{a}{N}>0$, the multiple defaults condition written as 
%\be\label{crash-condition}
%\sup_{x\in [0,\infty)^N}|x^{h_i}-x^j|<\frac{N}{ak(N-1)}\left(2-k\gamma\right),
%\en 
%for some $k\in\{1,\cdots,N\}$ indexes $(h_1,\cdots,h_k)\subset\{1,\cdots,N\}$ is not satisfied in general.  
%
%\item Regarding the individual default, according to Proposition 2.3 in \cite{Fouque-Ichiba}, the condition  
%\[
%\inf_{x\in[0,\infty)} a(\overline x-x^i)\geq \frac{2-\gamma}{N-1}
%\]
%holds for an individual bank $i$  leading to the drift of $X^i$  than $2$. Hence, bank $i$ is never broke.  

\item Given a identical growth rate $\gamma$ and a symmetric lending preference $a/N$, the stochastic stability of the process  $X^i_t-\overline X_t$ for all $t\geq 0$ is guaranteed, Under this condition, the matrix-valued process $(X^i_t-X^j_t)_{1\leq i,j\leq N}$ is also stochastically stable using Proposition 2.4 and Corollary 2.1 in \cite{Fouque-Ichiba}.     
\end{itemize}

\begin{figure}
\includegraphics[width=15cm,height=6cm]{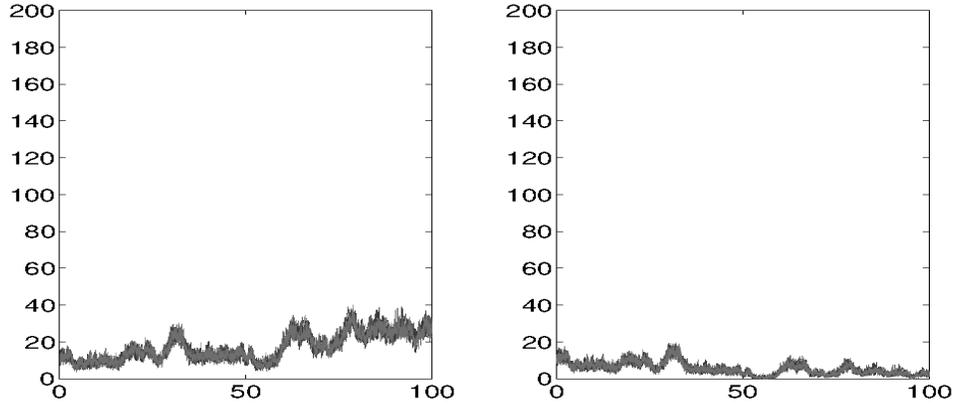}
\caption{One realization of $N=10$ trajectories of \eqref{coupled-no-game} with $\gamma=0.2$ (left plot) and with $\gamma=0$ (right plot). The common mean-reverting rate $a$ is fixed at $10$.}
\label{realization-02}
\end{figure}

\begin{figure}
\includegraphics[width=15cm,height=6cm]{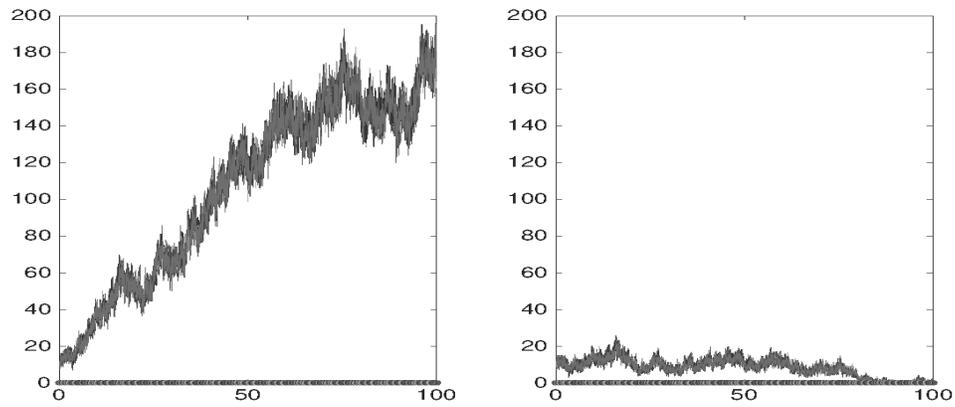}
\caption{One realization of $N=10$ trajectories of \eqref{coupled-no-game} with $\gamma=2$ (left plot) and with $\gamma=0$ (right plot). The common mean-reverting rate $a$ is fixed at $10$.}
\label{realization-2}
\end{figure}

%In particular, if the growth rate $\gamma$ is large enough, we observe that not only system crash never occurs but also an individual bank never defaults. 

%\begin{prop} 
%In the case of growth rate $\gamma>2$, an individual bank is never broke. 
%\end{prop}
%\begin{proof}
%Applying the comparison theorem in \cite{Ikeda-Watanabe1977} with the squared Bessel processes and using $\frac{1}{N}Y_t\leq X^i_t$, we obtain that the process
%\[
%X^i_t=\xi^i+\int_0^t\left(a(\frac{1}{N}Y_s-X^i_s)+\gamma\right)ds+2\int_0^t\sqrt{X^i_s}dW^i_s,
%\]
%  is less than or equal to the squared Bessel process given by
%\ban
%\check X^i_t=\xi^i+\int_0^t\left(a(\frac{1}{N}X^i_s-X^i_s)+\gamma\right)ds+2\int_0^t\sqrt{\check X^i_s}dW^i_s,
%\ean
%  

%\end{proof}

 Here, we apply the Euler scheme to plot a typical realization of $N=10$ trajectories of (\ref{coupled-no-game}) up to time $T=100$ with a time step $\Delta=10^{-4}$ and $a$ fixed at $10$. The growth rate $\gamma$ is varied as $0.2$ and $0$ in Figure \ref{realization-02} and as $2$ and $0$ in Figure \ref{realization-2}. Note that the trajectories are grouped by the mean reverting rate $a=10$. This flocking behavior leads to stability at the larger growth rate creates stability but systemic risk at the smaller growth rate. Figure \ref{realization-02} shows that due to $\gamma=0.2$, the growth rate of total monetary reserve $Y_t$ is $2$ satisfying the condition (\ref{cond-3}) so that we observe the shrinking scenario in the sense that $Y_t$ becomes small in this particular realization. The monetary reserve $X^i_t$ can be rewritten as 
\ban
d X^{i}_t &=& \left(\frac{a}{N}\sum_{j=1}^N(  X^j_t- X^i_t)+\gamma\right)dt+2\sqrt{X^i_t}d  {W}^{i}_t\\
&=&\bigg(a(\overline{X}_t-X^i_t)+\gamma\bigg)+2\sqrt{X^i_t}d  {W}^{i}_t\,, \quad  i=1,\cdots, N,
\ean
where $X^i_t$ is mean-reverting at the averaged capitalization $\frac{1}{N}Y_t$. The previous discussion implies that lending and borrowing behaviors create stability when $\gamma>\frac{2}{N}$ but incur systemic risk when $0\leq\gamma\leq \frac{2}{N}$. 

%Hence, according to the above discussion, we observe that the growth rate acts as stabilizer for the system. A regulator must pay attention the behavior of the growth rate to prevent system crash. 

In addition, similarly to the model in \cite{R.Carmona2013}, the number of banks $N$ plays an important role in the system. In coupled OU processes, a large $N$ mitigates systemic risk since the probability of the leading averaged capitalization reaching the default barrier (i.e. entering a systemic event) decays with order $\sqrt{N}$. The systemic event triggering a system crash occurs independently of the mean-reverting rate $a$. On the other hand,  in interacting Feller diffusions, the system crash occurs through the total monetary reserve of the interacting Feller diffusions driven by $N\gamma$. The probability of all defaults diminishes meaning that the survival probability of the system increases under a strong growth rate $\gamma$ and a large $N$. Note that this probability is independent of the lending preference $a/N$.  
%However, the lending preference still affects the facility of the avoidance of multiple defaults. Given the fixed growth rate $\gamma$, the conditions \eqref{crash-condition} and \eqref{survival-condition} show that a system with the high frequent lending and borrowing given by a large rate $a$ has slight possibility to face multiple defaults. 
Hence, in the case of the coupled system with an identical lending preference, for the propose of preventing systemic risk or multiple defaults, a regulator must expend the number of individuals and encourage lending and borrowing behavior.

\section{Construction of Nash Equilibria}\label{Exact-Nash}
We now return to the model with a game feature proposed in Section \ref{Intro}. The monetary reserve of bank $i$ is written as
\be\label{coupled}
d X^{i}_t = \left(a(\overline X_t-X^i_t)+\gamma_t\right)dt+\alpha^i_t dt
 +2\sqrt{X^i_t}d  {W}^{i}_t\,, \quad  i=1,\cdots, N,
\en
 where the growth rate nonnegative  $\gamma_t $  is a deterministic function in $L_{\infty}$ space and the initial value $X^i_0=\xi^i$ may be a nonnegative square-integrable and identically distributed random variable independent of Brownian motions. Each bank attempts to minimize 
\ba\label{objective}
\EE\bigg\{\int_{0}^{T} f_i(X_t,\alpha^i_t)
dt+g_i(X_{T}^{i})\bigg\},
\ea
where $X=(X^1,\cdots,X^N)$, $x=(x^1,\cdots,x^N)$, $\alpha=(\alpha^1,\cdots,\alpha^N)$ and 
\[
f_i(X_t,\alpha^i_t)
=\frac{(\alpha^i_t)^2}{2}-q\alpha^i_t(\overline{X}_t-X^i_t)
+\frac{\epsilon}{2}(\overline{X}_t-X^i_t)^2,
%+\frac{(\alpha^i_{t-\delta})^2}{2}
%+\frac{\tilde\epsilon}{2}(\overline{X}_{t-\delta}-X^i_{t-\delta})^2
\]
and $q^2\leq\epsilon$ ensuring that $f_i(x,\alpha)$ is convex in $(x,\alpha)$ 
and
\[
g_i(X_T)=\frac{c}{2}\left(\overline{X}_T-X^{i}_T\right)^2.
\]

%\subsection{Markov Nash Equilibria in Finite Players Games}
\subsection{HJB Approach}\label{HJB-approach}
Referring to the discussion in Chapter 5 of \cite{CarmonaSIAM2016}, in the Markovian setting, the Markov Nash equilibrium can be explicitly solved via the HJB equations. As banks intend to minimize their costs, given the optimal strategies $\hat\alpha^j$ for $j\neq i$ with the corresponding trajectories $\hat X^{-i}=(\hat X^1,\cdots,\hat X^{i-1},\hat X^{i+1},\cdots,\hat X^N)$, the value function of bank $i$ is written as
\be\label{value-function}
V^i(t,x)=\inf_{\alpha^i\in{\cal{A}}^i}\EE_{t,x}\left\{\int_t^T f_i(X^i_s, \hat X^{-i}_s, \alpha^i_s)ds+g_i(X_{T}^{i},\hat X^{-i}_T)\right\}.
\en
Hence, the HJB equation for solving for the Markov Nash equilibrium suggested by the dynamic programming principle reads
%\ba\label{HJB}
%\nonumber -\gamma V^i&+&\inf_{\alpha^i}\bigg\{
%\sum_{j\neq i}\bigg(a(\overline{x}-x^j)+{\alpha^j(t,x)}\bigg)\partial_{x^j}V^i
%+ \bigg(a(\overline{x}-x^i)+{\alpha^i(t,x)}\bigg)\partial_{x^i}V^i\\
%&+&+\frac{1}{2}\sum_{j=1}^N   x^j \partial_{x^jx^j}V^i\frac{(\alpha^{i})^2}{2}-q\alpha^{i}\left(\overline{x}-x^i\right)+\frac{\epsilon}{2}(\o x-x^i)^2\bigg\} =0,
%\ea
\ba\label{HJB}
\nonumber  \partial_{t}V^i&+&\inf_{\alpha^i}\bigg\{
\sum_{j\neq i}\bigg(a(\overline{x}-x^j)+\gamma_t+{\hat\alpha^j(t,x)}\bigg)\partial_{x^j}V^i
+ \bigg(a(\overline{x}-x^i)+\gamma_t+{\alpha^i}\bigg)\partial_{x^i}V^i\\
&+&2 \sum_{j=1}^N   x^j \partial_{x^jx^j}V^i+\frac{(\alpha^{i})^2}{2}-q\alpha^{i}\left(\overline{x}-x^i\right)+\frac{\epsilon}{2}(\overline x-x^i)^2\bigg\} =0,
\ea
with the terminal condition $V^i(T,x)=\frac{c}{2}(\overline x -x^i)^2$. Here, $\hat\alpha^j$ for all $j\neq i$ denote the optimal strategies of bank $j$ for all $j\neq i$.   We minimize (\ref{HJB}) with respect to $\alpha^i$ leading to the resulting strategy for player $i$  
\ba\label{optimal-closed} 
\hat\alpha^i=q^i(\overline x-x^i)-\partial_{x^i}V^i.
\ea
Inserting the optimal strategy candidate \eqref{optimal-closed} into the HJB equation \eqref{HJB}, \eqref{HJB} becomes 
\ba
\nonumber \partial_tV^i&+&\sum_{j=1}^N\bigg((a+q)(\overline{x}-x^j)-\partial_{x^j}V^j+\gamma_t\bigg)\partial_{x^j}V^i+2\sum_{j=1}^Nx^j \partial_{x^jx^j}V^i\\
&+&\frac{(\partial_{x^i}V^i)^2}{2}+\frac{(\epsilon-q^2)}{2}(\overline{x}-x^i)^2=0.
\ea
We now make the following ansatz for $V^i$ 
\be\label{ansatz-Vi}
V^i(t,x)=\frac{\eta^c_t}{2}(\overline{x}-x^i)^2+L^c_t(\overline{x}-x^i)+\phi^c_t\overline x+\mu^c_t,
\en
where $\eta^c_t$, $L^c_t$, $\phi^c_t$ and $\mu^c_t$ are deterministic functions. Inserting \eqref{ansatz-Vi} into \eqref{HJB} and identifying the terms $(\overline{x}-x^i)^2$, $(\overline{x}-x^i)$, $\overline x$, and the state-independent terms, we obtain a system of ordinary differential equations written as
\ban
 \dot\eta^c_t&=&2(a+q)\eta^c_t+(1-\frac{1}{N^2})(\eta^c_t)^2-(\epsilon-q^2),\\
 \dot L^c_t&=&\left(a+q+ \frac{1}{N}(1-\frac{1}{N})\eta^c_t\right)L^c_t+\frac{1}{N}(1-\frac{1}{N})\eta^c_t\phi^c_t+2(1-\frac{2}{N})\eta^c_t,\\
 \dot\phi^c_t&=&-2(1-\frac{1}{N})\eta^c_t,\\
 \dot\mu^c_t&=&\frac{1}{N}(1-\frac{1}{2N})(\phi^c_t)^2-\frac{1}{2}(1-\frac{1}{N})^2(L^c_t)^2-(1-\frac{1}{N})^2L^c_t\phi^c_t-\gamma_t\phi^c_t,
\ean
with terminal conditions $\eta^c_T=c$, $L^c_T=0$, $\phi^c_T=0$, and $\mu^c_T=0$. Solving this system, we get 
\be\label{eta-explicit}
\eta^c_{t}  = \frac{-(\epsilon-q^2)\left(e^{(\delta^+-\delta^-)(T-t)}-1\right)-c\left(\delta^+ e^{(\delta^+-\delta^-)(T-t)}-\delta^-\right)}
{\left(\delta^-e^{(\delta^+-\delta^-)(T-t)}-\delta^+\right)-c(1-\frac{1}{N^2})\left(e^{(\delta^+-\delta^-)(T-t)}-1\right)},
\en
where we use the notation 
\be\label{deltas}
\delta^{\pm}=-(a+q)\pm\sqrt{R},
\en
with
\be\label{eq:R}
R:=(a+q)^{2}+\left(1-\frac{1}{N^2}\right)(\epsilon-q^{2})>0,
\en
\be\label{L}
L^c_t=\int_t^Te^{\int_s^t\left(a+q+ \frac{1}{N}(1-\frac{1}{N})\eta^c_u\right)du}\left(\frac{1}{N}(\frac{1}{N}-1)\eta^c_s\phi^c_s+2(\frac{2}{N}-1)\eta^c_s\right)ds,
\en
\be\label{phi}
\phi^c_t=\int_t^T2(1-\frac{1}{N})\eta^c_sds,
\en
and
\be
\mu_t^c=-\int_t^T\left\{\frac{1}{N}(1-\frac{1}{2N})(\phi^c_s)^2-\frac{1}{2}(1-\frac{1}{N})^2(L_s^c)^2- (1-\frac{1}{N})^2L^c_s\phi^c_s-\gamma_s\phi^c_s\right\}ds.
\en
Observe that $\eta_t$ is well defined for any $t\leq T$  given by
the negative denominator of (\ref{eta-explicit})  
\ban
&&-\left(e^{(\delta^+-\delta^-)(T-t)}+1\right)\sqrt{R}
-\left(a+q+c\left(1-\frac{1}{N^2}\right)\right)\left(e^{(\delta^+-\delta^-)(T-t)}-1\right)<0,
\ean
owing to $\delta^+-\delta^-=2\sqrt{R}>0$ studied in \cite{R.Carmona2013}. Therefore, using (\ref{ansatz-Vi}), the corresponding optimal strategy for bank $i$ is written as  
\begin{equation}\label{optimal-finite-ansatz}
\hat\alpha^i_t=\left(q+(1-\frac{1}{N})\eta^c_t\right)(\overline X_t-X^i_t)-\psi^c_t,
\end{equation}
which is also the Markov Nash equilibrium. We write the corresponding dynamic $X^i_t$ for $i=1,\cdots,N$ as
\be\label{X-optimal}
d X^i_t=\bigg\{\left(a+q+(1-\frac{1}{N})\eta^c_t\right)(\overline {  X}_t- X^i_t)+\gamma_t-\psi^c_t\bigg\}dt +2 \sqrt{  X^i_t}dW^i_t, 
\en
where $\psi^c_t=(\frac{1}{N}-1)L^c_t+\frac{1}{N}\phi^c_t\geq 0$ expresses the deposit rate. Consequently, the total monetary reserve $ Y_t=\sum_{i=1}^N X^i_t$  is obtained as 
\be\label{sum}
dY_t=N(\gamma_t-\psi^c_t) dt+2\sqrt{  Y_t}d\widetilde W_t,
\en
with a standard Brownian motion $\widetilde W_t$ in some extension probability space.
In order to satisfy the regularity condition of $\eqref{sum}$ and guarantee the admissibility of $\hat\alpha^i$ for $i=1,\cdots,N$, we further assume 
\be\label{regularity-cond-1}
\gamma_t\geq \psi^c_t
\en
and 
\be\label{regularity-cond-2}
a+q+(1-\frac{1}{N})\eta_t^c\leq N,
\en
for $0\leq t\leq T$ leading to the drift coefficient of $X$ bounded below by $\gamma_t-\psi^c_t\geq 0$ for $0\leq t\leq T$. Based on Proposition 2.1 in \cite{Fouque-Ichiba}, the existence and uniqueness of the weak solution to $ X^i$ for $i=1,\cdots,N$ given by (\ref{X-optimal}) in the positive polyhedral domain $[0,\infty)^N$ can be obtained. Hence, we verify the admissibility of the Markov Nash equilibrium $\hat\alpha^i$ for $i=1,\cdots,N$.   

The financial interpretation of the assumption (\ref{regularity-cond-1}) is that banks are forbidden from leaving more deposits than their monetary growth through an appropriate incentive $q$ imposed by a regulator. The condition (\ref{regularity-cond-2}) mentions that the number of banks in the interbank lending system is large enough. In order to verify that $V^i$ given by (\ref{ansatz-Vi}) is the value function and the corresponding optimal strategy is $\hat\alpha^i$ written as ({\ref{optimal-finite-ansatz}}), we study the verification theorem given by Theorem \ref{Ver-Thm} in Appendix \ref{Appex-1}.

Below, we study the $\epsilon$-Nash equilibria and defer the financial implications to Section \ref{F-I}.

%{\color{blue}
\subsection{$\epsilon$-Nash Equilibria: Mean Field Games}
In MFGs, we solve the $\epsilon$-Nash equilibria in the mean field limit $N\rightarrow\infty$. According to the steps in \cite{R.Carmona2013}, we first fix $(m_t)_{t\geq 0}$ as a candidate for the limit $\overline X_t$ as $N\rightarrow\infty$  
\be\label{limit-xbar}
m_t=\lim_{N\rightarrow\infty}\overline X_t,
\en
and solve the one-player standard control problem 
\be\label{value-MFG}
V^m(t,x)=\inf_{\alpha_t}\EE_{t,x}\bigg\{\int_0^T\left[\frac{\alpha^2_t}{2}-q\alpha_t(m_t-X_t)+\frac{\epsilon}{2}(m_t-X_t)^2\right]dt+\frac{c}{2}(m_T-X_T)^2\bigg\},
\en
subject to 
\be\label{X0-MFG}
dX_t= \bigg(a(m_t-X_t)+\gamma_t\bigg)dt+\alpha_tdt+2\sqrt{X_t}dW_t,
\en
with the nonnegative  deterministic growth rate  $\gamma_t $ in $L_{\infty}$ space where $W_t$ is a standard Brownian motion independent of the initial value $X_0=\xi$ which may be a nonnegative square integrable random variable. Finally, we obtain $m_t$ such that $m_t=\EE[X_t]$ for all $0\leq t\leq T$. Given the candidate $m_t$,  the corresponding HJB equation to the problem (\ref{value-MFG}-\ref{X0-MFG}) is written as 
\be\label{HJB-MFG}
\partial_tV^m+\inf_{\alpha}\bigg\{\left(a(m_t-x)+\gamma_t+\alpha\right)\partial_xV^m+2x\partial_{xx}V^m+\frac{\alpha^2}{2}-q\alpha(m_t-x)+\frac{\epsilon}{2}(m_t-x)^2\bigg\}=0.
\en
The first order condition leads to optimal strategy given by
\be\label{optimal-MFG}
\hat\alpha_t=q(m_t-X_t)-\partial_xV^m. 
\en
Plugging (\ref{optimal-MFG}) into (\ref{HJB-MFG}) implying
\be
\partial_tV^m+\left((a+q)(m_t-x)-\partial_xV^m+\gamma_t\right)\partial_xV^m+2x\partial_{xx}V^m+\frac{(\partial_xV^m)^2}{2} +\frac{\epsilon-q^2}{2}(m_t-x)^2=0,
\en 
and then using the ansatz $V^m(t,x)=\eta^m_t(m_t-x)+L^m_t(m_t-x)+\phi^m_tm_t+\mu^m_t$, we obtain that $\eta^m_t$, $L^m_t$, and $\mu^m_t$ must satisfy
\ban
\label{eta-MFG}\dot\eta^m_t&=&2(a+q)\eta^m_t+(\eta^m_t)^2-(\epsilon-q^2),\\
 \dot L^m_t&=&(a+q)L^m_t+2 \eta^m_t,\\
 \dot\phi^m_t&=&-2 \eta^m_t, \\
\label{mu-MFG}\dot\mu^m_t&=&-\frac{1}{2}(L^m_t)^2-L^m_t\phi^m_t-\gamma_t\phi^m_t,
\ean
with terminal conditions $\eta^m_T=c$, $L^m_T=0$, $\phi^m_T=0$, and $\mu^m_T=0$ leading to 
\be\label{eta-explicit-MFG}
\eta^m_{t}  =\eta_t= \frac{-(\epsilon-q^2)\left(e^{(\delta^+-\delta^-)(T-t)}-1\right)-c\left(\delta^+ e^{(\delta^+-\delta^-)(T-t)}-\delta^-\right)}
{\left(\delta^-e^{(\delta^+-\delta^-)(T-t)}-\delta^+\right)-c\left(e^{(\delta^+-\delta^-)(T-t)}-1\right)},\nonumber
\en
where we again use the notation  
\be\label{deltas-m}
\delta^{\pm}=-(a+q)\pm\sqrt{R},\nonumber
\en
with
\be\label{eq:R-m}
R:=(a+q)^{2}+(\epsilon-q^{2})>0,\nonumber
\en
\be\label{L-m} 
L^m_t=L_t=-2\int_t^Te^{\left(a+q\right)(t-s)}  \eta^m_s ds,\nonumber
\en
\be\label{phi-m}
\phi^m_t=\phi_t 2\int_t^T\eta^m_sds,\nonumber
\en
and
\be
\mu^m_t=\mu_t=\int_t^T\left\{\frac{1}{2}(L^m_s)^2+ L^m_s\phi^m_s+\gamma^m_s\phi^m_s\right\}ds.\nonumber
\en
Hence, the $\epsilon$-Nash equilibrium is given by
 \begin{equation} 
 \hat\alpha_t=(q+\eta_t)(m_t-X_t)-\psi_t,
 \end{equation}
with deposit rate $\psi_t=- L_t\geq 0$ for $0\leq t\leq T$ and the corresponding dynamic $X_t$ is driven by
\[
d X_t=((a+q+\eta_t)(m_t- X_t)-\psi_t+\gamma_t)dt+2\sqrt{  X_t}dW_t.
\]
We obtain that $m_t=\EE(X_t)$ for $0\leq t\leq T$ written as
\[
dm_t=(\gamma_t-\psi_t)dt. 
\] 
where we assume  $\gamma_t-\psi_t\geq 0$ for all $0\leq t\leq T$. The $\epsilon$-Nash equilibrium \eqref{optimal-MFG} is the exact limit of the Markov equilibrium in the finite player games \eqref{optimal-finite-ansatz} as $N\rightarrow\infty$. 
Hence, we naturally regard the $\epsilon$-Nash equilibrium as an asymptotic solution of the optimal strategy for bank $i$ written as 
\[
\hat\alpha^i_t=(q+\eta_t)(\overline X_t-X^i_t)-\psi_t,
\] 
satisfying $a+q+\eta_t\leq N$ for all $0\leq t\leq T$ in order to guarantee admissibility referring to Proposition 2.1 in \cite{Fouque-Ichiba}. The mean field limit $m_t$ simply is replaced from the averaged capitalization $\overline{X}_t$. 
%}

\begin{remark}
The Markov Nash equilibrium and the $\epsilon$-Nash equilibrium are also obtained using the Dynamic programming-Carath\'eodory approach. In fact, since   the HJB approach gives the equilibrium explicitly in Section \ref{HJB-approach}, the Dynamic programming-Carath\'eodory approach is treated as a verification theorem. However, in the non-Markovian case, we apply this approach to solve the equilibrium. See \cite{alekal1971quadratic}, \cite{L.Chen2011}, \cite{huang2012forward}, and \cite{Carmona-Fouque2016} for instance. 
\end{remark}

\subsection{Financial Implications}\label{F-I}
This subsection is devoted to analysis of the model (\ref{coupled}-\ref{objective}). 
We use the Markov Nash equilibria identical to the $\epsilon$-Nash equilibria. 
Given the strategy of bank $i$,
\be\label{optimal-FI}
\alpha^i_t=(q+(1-\frac{1}{N})\eta^c_t)(\overline X_t-X^i_t)-\psi^c_t, 
\en
the dynamics of bank $i$ are expressed as 
\ba
\nonumber dX^i_t&=&\bigg\{(a+q+(1-\frac{1}{N})\eta^c_t)(\overline X_t-X^i_t)+\gamma_t-\psi^c_t\bigg\}dt+2\sqrt{X_t^i}dW^i_t\\
&=&\left\{\frac{a+q+(1-\frac{1}{N})\eta_t^c}{N}\sum_{j=1}^N(X^j_t-X^i_t)+\gamma_t-\psi^c_t \right\}dt+2\sqrt{X_t^i}dW^i_t,\label{Im-coupled}
\ea
and the total monetary reserve $Y_t=\sum_{j=1}^N X^i_t$ is 
\begin{equation}
dY_t=N(\gamma_t-\psi^c_t )dt+2\sqrt{  Y_t}d\widetilde W_t,
\end{equation}
where $\widetilde W_t$ is a standard Brownian motion in some extension probability space. 

\begin{itemize}

%\begin{figure}
%\includegraphics[width=15cm,height=6cm]{etaandAvsqwithtitle.eps}
%\caption{Plots of $\eta_t$ on the left and $A_t$ on the right for $0\leq t\leq T$ with varied $q$ and $a=1$, $\epsilon=2$, $c=0$, $N=10$, and $T=1$.}
%\label{etavsq}
%\end{figure}

\item According to the analysis in Section \ref{Interbank}, we imposes the instability condition 
\[
\sup_{0\leq t\leq T}(\gamma_t-\psi^c_t )<\frac{2}{N},
\] 
and discuss the tail probability of all defaults based on the total monetary reserve $Y_t$ and its corresponding  first passage time $\tau=\inf\{t> 0;Y_t=0\}$. 
%written as $U(t,y,\gamma,\psi^c)=\PP(\tau > T)$. Denote $\widetilde\gamma_t=\gamma_t-\psi^c_t$ such that $U(t,y,\widetilde\gamma)$  satisfying the PDE
%\begin{equation*}
%\partial_tU+\widehat\gamma\partial_{\widetilde\gamma} U+\widetilde\gamma\partial_y U+\frac{1}{2}y\partial_{yy}U=0,
%\end{equation*}
%with the terminal condition $U(T,y,\widetilde\gamma)=1$ and the boundary condition $U(t,0,\widetilde\gamma)=0$ where 
%\[
%\widehat\gamma=\partial_t\gamma_t+(1-\frac{1}{N})\left((a+q+\frac{1}{N}(1-\frac{1}{N})\eta^c_t)L^c_t+(1-\frac{1}{N})\eta^c_t(\frac{1}{N}\phi^c_t+1)\right).
%\]
Due to $Y_t$ driven by the time-dependent drift, the explicit solution of tail probability $\PP(\tau > T)$ is not obtained in the literature. Therefore, we consider the estimation of the tail probability of all defaults. To this end, define $\hat\tau=\inf\{t>0;\hat Y_t=0\}$ and $\check\tau=\inf\{t> 0;\check Y_t=0\}$ with corresponding dynamics 
\[
d\hat Y_t=N\sup_{0\leq t\leq T}(\gamma_t-\psi^c_{t}) dt +2\sqrt{\hat Y_t}d\widetilde W_t,
\] 
and 
\[
d\check Y_t=N\inf_{0\leq t\leq T}(\gamma_t-\psi^c_t) dt +2\sqrt{\check Y_t}d\widetilde W_t.
\] 
%with the same initial value $Y_0=\sum_{j=1}^N X^j_0$ where $\hat t $ and $\check t$ are defined as
%\ba
%\hat t=\sup_{0\leq t\leq T}\{t:\gamma_t-\psi^c_t\},\\
%\check t=\inf_{0\leq t\leq T}\{t:\gamma_t-\psi^c_t\}.
%\ea
Applying the comparison theorem studied in \cite{Yor2003, Ikeda-Watanabe1977, Li-Yor1999}, we estimate the tail probability of all defaults as 
\ba
\nonumber \Gamma\left(\frac{Y_0^2}{2T};N\inf_{0\leq t\leq T}(\gamma_t-\psi^c_{ t}) \right)&=&\PP(\check\tau>T)\\
\nonumber &\leq&\PP(\tau>T)\\
&\leq&\PP(\hat\tau>T)=\Gamma\left(\frac{Y_0^2}{2T};N\sup_{0\leq t\leq T}(\gamma_t-\psi^c_{t})\right),
\ea
for $T\geq t$ where $\Gamma(x,s):=\int_0^xu^{s-1}e^{-u}du$ for $s>0$, and $x\geq 0$.  

%\item Denoting $A_t=a+q+(1-\frac{1}{N})\eta_t$, we also study the estimation of tail probability of multiple defaults under the condition written as 
%\be 
%\sup_{x\in [0,\infty)^N}|x^{h_i}-x^j|<\inf_{0\leq t\leq T}\bigg\{\frac{N}{A_tk(N-1)}\left(2-k(\gamma_t-\psi^c_t)\right)\bigg\},
%\en
%for some $k\in\{1,\cdots,N\}$ indexes $(h_1,\cdots,h_k)\subset\{1,\cdots,N\}$. Given $A_t\leq N$ for all $0\leq t\leq T$,   comparing the sum of these $k$ banks $Y^k_t=\sum_{i=1}^kX^{h_i}_t$ to the squared Bessel processes $\hat Y^k_t$ with dimension 
%\[
%\overline\gamma=\sup_{0\leq t\leq T}\bigg\{k(\gamma_{ t}-\psi^c_{ t}) +\sup_{x\in[0,\infty)}\frac{A_t}{N}\left|\sum_{i=1}^k\sum_{j=1}^N(x_j-x_{h_{i}})\right|\bigg\},  
%\]
%and $\check Y^k_t$ with dimension
%\[
%\underline\gamma=\inf_{0\leq t\leq T}\bigg\{k(\gamma_{  t}-\psi^c_{  t})+\inf_{x\in[0,\infty)}\frac{A_t}{N}\sum_{i=1}^k\sum_{j=1}^N(x_j-x_{h_i})\bigg\},
%\]
%by the comparison theorem again, we obtain the rough lower and upper bounds for the tail probability of multiple default time $\tau^k=\inf\{t> 0;Y^k_t=0\}$ given by 
%\be
%\Gamma\left(\frac{(Y^k_0)^2}{2T};\underline\gamma\right)=\PP(\check\tau^k>T)\leq\PP(\tau^k>T)\leq\PP(\hat\tau^k>T)=\Gamma\left(\frac{(Y^k_0)^2}{2T};\overline\gamma\right);\quad t\geq 0,
%\en
%where $\check\tau^k$ and $\hat\tau^k$ are the first passage time of $\check Y^k_t$ and $\hat Y^k_t$ written as 
%\[
%\check\tau^k=\inf\{t> 0;\check Y^k_t=0\},\quad \hat\tau^k=\inf\{t> 0;\hat Y^k_t=0\}. 
%\]

\item Though the rebalanced growth rate $\gamma_t-\psi^c_t$ and the symmetric lending preference $A_t/N$ are not purely constants but deterministic homogeneous functions where 
\be\label{increasing-liquidity}
A_t=a+q+(1-\frac{1}{N})\eta_t^c,
\en 
we still apply Proposition 2.4 and Corollary 2.1 in \cite{Fouque-Ichiba}. The stochastic stability of process $X^i_t-\overline X_t$ and the matrix-valued process $(X^i_t-X^j_t)_{1\leq i,j\leq N}$ for all $0\leq t \leq T$ can be qualified through the identity of $\gamma_t-\psi^c_t$ and $A_t/N$.

\end{itemize}

\begin{figure}
\includegraphics[width=15cm,height=6cm]{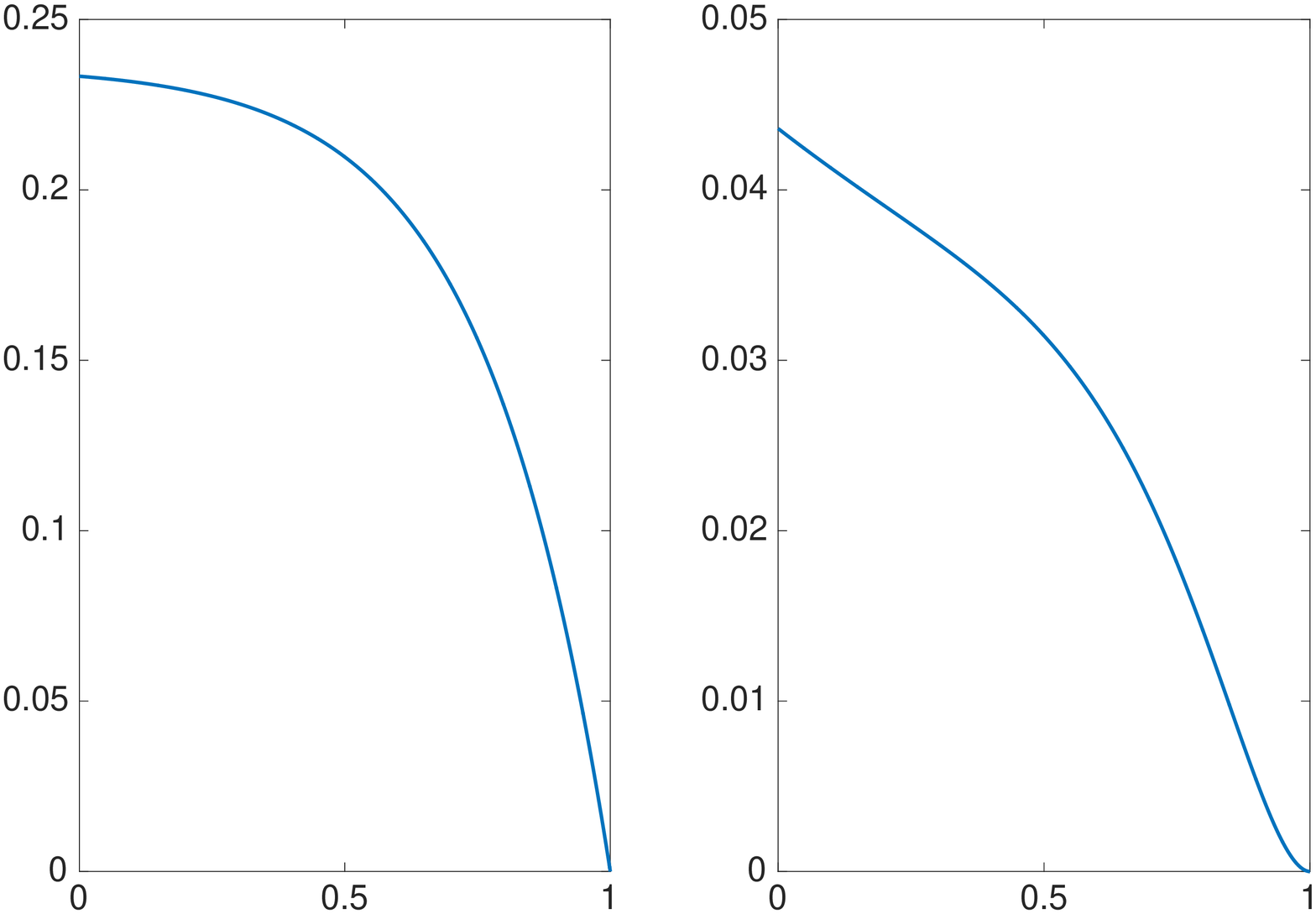}
\caption{Plots of $\eta_t$ on the left and $\psi_t$ on the right for $0\leq t\leq T$ with $a=1$, $q=1$, $\epsilon=2$, $c=0$,  $N=10$, $T=1$.}
\label{etavspsi}
\end{figure}

\begin{figure}
\includegraphics[width=15cm,height=6cm]{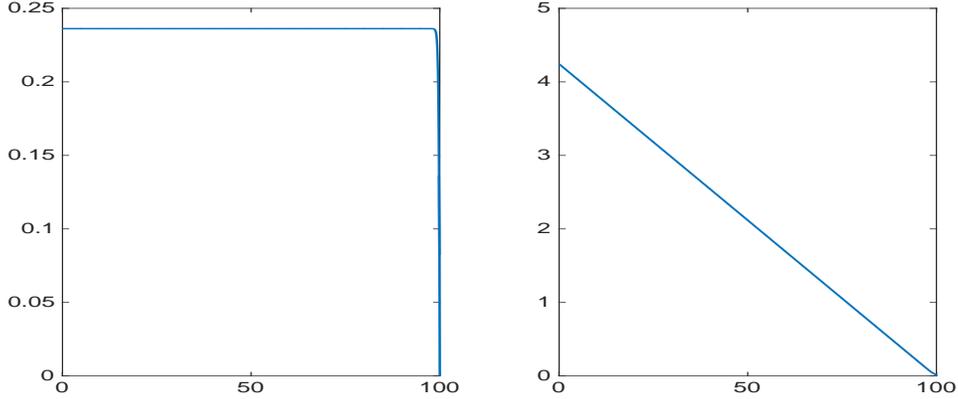}
\caption{Plots of $\eta_t$ on the left and $\psi_t$ on the right for $0\leq t\leq T$ with $a=1$, $q=1$, $\epsilon=2$, $c=0$,  $N=10$, $T=100$.}
\label{etavspsi100}
\end{figure}

As the results obtained in \cite{R.Carmona2013},  the objective function \eqref{objective} yields an increasing liquidity qualified by the larger rate of lending and borrowing given by (\ref{increasing-liquidity}). The fully symmetric lending and borrowing model leads to the Nash equilibrium independent of the growth rate $\gamma_t$ but according to the discussion in Section \ref{Interbank}, the growth rate $\gamma_t$ plays an important role as a stabilizer for the system. 
In particular, due to the volatility driven by the square root function of states written as $\sqrt{X^i_t}$,  under the optimal strategy \eqref{optimal-FI}, bank $i$ intends to borrow less money from the central bank or to lend it more money. In this way, it manipulates its own reserve through the positive deposit rate appearing in the drift as shown in Figure \ref{etavspsi} and  \ref{etavspsi100} in order to minimize the quadratic cost with a smaller volatility. Note that when $T$ is large, $\eta_t^c$ is treated as a constant and $\psi_t^c$ increases almost linearly with $t$.  Hence, the central bank acts as a central deposit corporation rather than a clearing house. However, noting that the growth rate diminishes through the deposit rate $\psi^c_t$, referring to the results in \cite{Fouque-Ichiba}, we observe that the equilibrium creates the possibility of all defaults. Consequently, under the stability condition 
\be\label{stability-cond}
\inf_{0\leq t\leq T}\left(\gamma_t-\psi^c_t\right)>\frac{2}{N}, 
\en
increasing liquidity leads to stability. However, if the system remains in the worse-case scenario given by 
\[
\sup_{0\leq t\leq T}\left(\gamma_t-\psi^c_t\right)\leq\frac{2}{N}, 
\]  
the interbank lending and borrowing system may face immediate systemic risk due to the large liquidity.  Therefore, from the regulator's point of view, the system is stable only when the deposit rate $\psi^c_t$ satisfies the stability condition \eqref{stability-cond}. Consequently, the regulator must control the incentive $q$ and the running penalty $\epsilon$  such that condition \eqref{stability-cond} and $\PP(Y_t=0,\; 0\leq t\leq T)=0$ hold meaning that impossibility of system crash in the fixed time period. In fact, a regulator may propose the restricted admissible set of the equilibrium written as
for $0\leq t\leq T$,
\be\label{admissible-cond}
X^{j,\alpha^i}_t\geq 0,\; \forall j=1,\cdots,N\;\mathrm{and}\; Y_t^{\alpha}>0\;a.s.,
\en  
leading to the weaker stability condition 
\be\label{stability-cond}
\inf_{0\leq t\leq T}\left(\gamma_t-\psi^c_t\right)\geq\frac{2}{N}.
\en
As an example, we discuss the MFG case $N\rightarrow\infty$ without lending and borrowing impositions and the terminal cost penalty that is, $a=c=0$, and  identify the sufficient conditions that satisfy \eqref{stability-cond}.

\begin{prop}\label{Prop-1}
Given the constant growth rate $\gamma_t=\gamma$, $N\rightarrow\infty$ and $a=c=0$, if the incentive $q$ satisfies
\[
\frac{\sqrt{2\epsilon}}{\sqrt{\gamma+2}}<q\leq\sqrt{\epsilon}, 
\]
 for fixed $\epsilon$ or if the running penalty $\epsilon$ satisfies
\[
q^2\leq\epsilon<\frac{\gamma+2}{2}q^2, 
\]
for fixed $q$, the stability condition \eqref{stability-cond} holds.
\end{prop}
\begin{proof}
As $N\rightarrow\infty$, the deposit rate becomes $\psi_t=-L_t$. Using
$$
\delta^+=-q+\sqrt{\epsilon},\ \delta^-=-q-\sqrt{\epsilon},
$$
and 
$$
\begin{array}{l}
\delta^+-\delta^- e^{(\delta^+-\delta^-)(T-t)}\\
=\delta^+-\delta^--\delta^- (e^{(\delta^+-\delta^-)(T-t)}-1)\\
=(q+\sqrt{\epsilon}) (e^{(\delta^+-\delta^-)(T-t)}-1)\\
\geq q (e^{(\delta^+-\delta^-)(T-t)}-1),
\end{array}
$$
we obtain
\ban
\eta_t&=&\frac{-(\epsilon-q^2)\left(e^{(\delta^+-\delta^-)(T-t)}-1\right)}
{\left(\delta^-e^{(\delta^+-\delta^-)(T-t)}-\delta^+\right)}\\
&\leq&\frac{\epsilon-q^2}{q}.
\ean 
and consequently
\ban
L_t&=&-2\int_t^Te^{q (t-s)}  \eta_s ds\\
&\geq&-2\frac{\epsilon-q^2}{q}\int_t^Te^{q (t-s)}ds\\
&=&2\frac{\epsilon-q^2}{q^2}\left(e^{-q(T-t)}-1\right) \\
 &\geq&2\left(1-\frac{\epsilon}{q^2}\right).
\ean
Since $\psi_t^c=-L_t$ and $\gamma-\psi^c_t=\gamma+L_t$ is bounded from below by $\gamma+2\left(1-\frac{\epsilon}{q^2}\right)$, a sufficient condition that guarantees condition \eqref{stability-cond} holds is given by 
\[
\gamma+2\left(1-\frac{\epsilon}{q^2}\right)> 0
\]
implying that for fixed $\epsilon$, the incentive $q$ satisfies
\[
\frac{\sqrt{2\epsilon}}{\sqrt{\gamma+2}}<q\leq\sqrt{\epsilon}, 
\]
or for fixed $q$, the running penalty $\epsilon$ satisfies
\[
q^2\leq\epsilon<\frac{\gamma+2}{2}q^2.
\]
\end{proof}
According to Proposition \ref{Prop-1}, given the running penalty $\epsilon$, a regulator must impose a minimum incentive $q$ of $\frac{\sqrt{2\epsilon}}{\sqrt{\gamma+2}}$ to satisfy \eqref{stability-cond} leading to the avoidance of system crash during the fixed time period.

\section{Nash Equilibria: Case of the Infinite Time Horizon}\label{Inf-Nash}
The aim of this section is to investigate the case of the infinite time horizon stochastic game. We extend the quadratic cost proposed by \cite{R.Carmona2013} in the infinite time horizon given a discount factor $e^{-rt}$ and a constant growth rate $\gamma\geq 0$ and study the behavior of the corresponding adding liquidity and deposit rate influenced by the discount rate $r$ and the incentive $q$ respectively. Adopting the notations in Section \ref{HJB-approach}, given the optimal strategies $\hat\alpha^j$ for $j\neq i$, the value function of bank $i$ is written as
\be\label{value-inf}
V^i(x)=\inf_{\alpha^i\in{\cal{A}}^i}\EE_{x}\left\{\int_0^\infty e^{-rt}f_i(X^i_s, \hat X^{-i}_s, \alpha^i_s)ds\right\}.
\en
where 
\[
f_i(X_t,\alpha^i_t)
=\frac{(\alpha^i_t)^2}{2}-q\alpha^i_t(\overline{X}_t-X^i_t)
+\frac{\epsilon}{2}(\overline{X}_t-X^i_t)^2,
\]
and $q^2\leq\epsilon$ so that $f_i(x,\alpha)$ is convex in $(x,\alpha)$ 
subject to the dynamics   
\be\label{coupled-inf}
d X^{i}_t = \left(\frac{a}{N}\sum_{j=1}^N(  X^j_t- X^i_t)+\gamma\right)dt+\alpha^i_t dt
 +2\sqrt{X^i_t}d  {W}^{i}_t\,, \quad  i=1,\cdots, N.
\en
with the initial $X^i_0=\xi^i$ identical to (\ref{coupled}). Similarly, in the Markovian setting, the HJB equation reads
\ba\label{HJB-inf}
\nonumber -rV^i&+&\inf_{\alpha^i}\bigg\{
\sum_{j\neq i}\bigg(a(\overline{x}-x^j)+\gamma+{\hat\alpha^j(t,x)}\bigg)\partial_{x^j}V^i
+ \bigg(a(\overline{x}-x^i)+\gamma+{\alpha^i}\bigg)\partial_{x^i}V^i\\
&+&2\sum_{j=1}^N   x^j \partial_{x^jx^j}V^i+\frac{(\alpha^{i})^2}{2}-q\alpha^{i}\left(\overline{x}-x^i\right)+\frac{\epsilon}{2}(\overline x-x^i)^2\bigg\} =0,
\ea
where $\hat\alpha^j$ for all $j\neq i$ represent the optimal strategies for bank $j$ for all $j\neq i$. The first order condition yields the strategy
\[
\hat\alpha^i_t=q(\overline X_t-X_t)-\partial_{x^i}V^i,
\]
and the solution to the value function $V^i$ is given by
\be\label{ansatz-inf}
V^i(x)=\frac{\eta}{2}(\overline x-x^i )^2+L(\overline x_t-x^i )+\phi\overline x +\mu,
\en
where $\eta$, $L$, $\phi$, and $\mu$ satisfy
\ban
 -r\eta&=&2(a+q)\eta+(1-\frac{1}{N^2})\eta^2-(\epsilon-q^2),\\
  - rL&=&\left(a+q+ \frac{1}{N}(1-\frac{1}{N})\eta\right)L+\frac{1}{N}(1-\frac{1}{N})\eta\phi+2(1-\frac{2}{N})\eta,\\
 -r\phi&=&-2(1-\frac{1}{N})\eta,\\
 -r\mu&=&\frac{1}{N}(1-\frac{1}{2N})\phi^2-\frac{1}{2}(1-\frac{1}{N})^2L^2-(1-\frac{1}{N})^2L\phi-\gamma\phi,
\ean
leading to the solutions
\be\label{eta-inf}
\eta=\frac{\epsilon-q^2}{(a+q+\frac{r}{2})+\sqrt{(a+q+\frac{r}{2})^2+(1-\frac{1}{N^2})(\epsilon-q^2)}},
\en
\be\label{phi-inf}
\phi=\frac{2}{r}(1-\frac{1}{N})\eta,
\en
\be\label{L-inf}
L=\frac{\frac{1}{N}(\frac{1}{N}-1)\eta\phi+2(\frac{2}{N}-1)\eta}{\left(a+q+r+ \frac{1}{N}(1-\frac{1}{N})\eta\right)},
\en
and
\be
\mu=\frac{\frac{1}{2}(1-\frac{1}{N})^2L^2+(1-\frac{1}{N})^2L\phi+\gamma\phi-\frac{1}{N}(1-\frac{1}{2N})\phi^2}{r}.
\en
Hence, the Markov Nash equilibrium strategy $\hat\alpha^i$ is  written as
\be\label{optimal-control-inf}
\hat\alpha_t^i=\left(q+(1-\frac{1}{N})\eta\right)(\overline X_t-X^i_t)-\psi. 
\en
The individual monetary reserve $X^i_t$ and the total monetary reserve $Y_t=\sum_{i=1}^N X^i_t$  are given by
\be
dX^i_t=\left\{\left(a+q+(1-\frac{1}{N})\eta\right)(\overline X_t-X^i_t)+\gamma-\psi\right\}dt+2\sqrt{X^i_t}dW^i_t.
\en
and 
\be\label{Y-inf}
dY_t=N(\gamma-\psi)dt+ 2\sqrt{  Y_t}d\widetilde W_t,
\en
where $\widetilde W_t$ is defined in Section \ref{Interbank}. Similarly, in order to guarantee the regularity condition of the total monetary reserve (\ref{Y-inf}) and the admissibility of the optimal strategies $\hat\alpha^i$ for $i=1,\cdots,N$, using Proposition 2.1 \cite{Fouque-Ichiba}, we assume 
\be\label{regularity-cond-1-inf}
0\leq\psi=(\frac{1}{N}-1)L+\frac{1}{N}\phi\leq \gamma,
\en
and 
\be\label{regularity-cond-2-inf}
a+q+(1-\frac{1}{N})\eta\leq N.
\en 
The corresponding financial interpretation is discussed in Section \ref{Exact-Nash}.  We verify that $V^i$ and $\hat\alpha^i$ is the solution to the problem (\ref{value-inf}-\ref{coupled-inf}) through Theorem \ref{Ver-Thm-inf} in Appedix \ref{Appex-1} .

Based on the results in Section \ref{Interbank}, we now consider the re-balanced growth rate $\gamma-\psi$ as a critical value to determine the possibility of system crash.  In the case of $\gamma-\psi>\frac{2}{N}$, systemic risk is always avoided. In the case of $\gamma-\psi=\frac{2}{N}$, banks will face a future crisis. In the case of $0<\gamma-\psi<\frac{2}{N}$, all banks will default at some large time and the total monetary reserve reflects at the point \{$0$\}. In the case of $\gamma-\psi=0$, all banks will default within a the finite time and remain at zero almost surely. 
 
In addition, owing to homogeneity of the value function \eqref{value-inf}, we obtain the identical growth rate $\gamma-\psi$ and the symmetric lending preference $A/N$ where $A=a+q+(1-\frac{1}{N})\eta$.  Similarly, Proposition 2.4 and Corollary 2.1 in \cite{Fouque-Ichiba} ensure stochastic stability of the process $X^i_t-\overline X_t$ and the matrix-valued process $(X^i_t-X^j_t)_{1\leq i,j\leq N}$ for all $t\geq 0$.

%\begin{figure}
%\includegraphics[width=15cm,height=6cm]{etaandAvsqinf.eps}
%\caption{Plots of $\psi$  with $q=1$(dash-line), $q=1.1$(dot-dash line), and $q=1.2$(solid line) using fixed $a=1$, $\epsilon=2$, $c=0$, $r=0.01$, and $N=10$.}
%\label{eta-A-inf}
%\end{figure}

\begin{prop}\label{prop-inf-1}
For fixed $a\geq 0$, $\epsilon\geq 0$, and the discount rate $r\geq 0$, given the incentive $0\leq q\leq \sqrt{\epsilon}$, the first derivative of $\eta$ with respect to $q$ is 
\begin{equation}\label{eta-q}
\partial_q\eta=\frac{-1+\frac{a+\frac{r}{2}+\frac{q}{N^2}}{\sqrt{(a+q+\frac{r}{2})^2+(1-\frac{1}{N^2})(\epsilon-q^2)}}}{1-\frac{1}{N^2}}< 0,\;0\leq q\leq \sqrt{\epsilon}.
\end{equation}
 \end{prop}
\begin{proof}
It suffices to show 
\[
{a+\frac{r}{2}+\frac{q}{N^2}}<{\sqrt{(a+q+\frac{r}{2})^2+(1-\frac{1}{N^2})(\epsilon-q^2)}},
\]
which is equivalent to show that
\[
\left({a+\frac{r}{2}+\frac{q}{N^2}}\right)^2<(a+q+\frac{r}{2})^2+(1-\frac{1}{N^2})(\epsilon-q^2).
\]
Consequently, we obtain
\ban
&&(a+q+\frac{r}{2})^2+(1-\frac{1}{N^2})(\epsilon-q^2)-\left({a+\frac{r}{2}+\frac{q}{N^2}}\right)^2\\
&=&\left(1-\frac{1}{N^2}\right)\left(\frac{1}{N^2}q^2+(2a+r)q+\epsilon\right)>0,
\ean
for $0\leq q\leq \sqrt{\epsilon}$ which completes the proof.
\end{proof}
For fixed $a\geq 0$, $\epsilon\geq 0$, and $q\geq 0$, $\eta$ given by \eqref{eta-inf} is a decreasing function of $r$.  According to \eqref{phi-inf}  and \eqref{L-inf}, we obtain 
\be\label{psi-inf}
\psi=(\frac{1}{N}-1)L+\frac{1}{N}\phi=\frac{\frac{1}{N}(\frac{1}{N}-1)(\frac{2}{N}-1)\eta\phi+2(\frac{1}{N}-1)(\frac{2}{N}-1)\eta+\frac{1}{N}(a+q+r)\phi}{a+q+r+\frac{1}{N}(1-\frac{1}{N})\eta}.
\en
Inserting $\phi$ into \eqref{psi-inf} and canceling $\eta$, $\psi$ becomes
\[
\psi=\frac{2(\frac{1}{N}-1)(\frac{2}{N}-1)\left(\frac{1}{rN}(1-\frac{1}{N})\eta^2+\eta\right)+\frac{2}{rN}(1-\frac{1}{N})(a+q+r)\eta}{a+q+r+\frac{1}{N}(1-\frac{1}{N})\eta}.
\]
By Proposition \ref{prop-inf-1},  a large $q$ creates a small deposit rate. In addition, observe that a large $r$ also suppresses the slight deposit rate. Therefore, a regulator must impose the restrictions on both the incentive and the discount rate to satisfy stability condition $\psi(q,r)<\gamma-\frac{2}{N}$ and hence facilitate the avoidance of all defaults. For example, assuming $N\rightarrow\infty$ with $a=0$, for any given $r$ and $\epsilon$, a regulator must impose a lower bound on the incentive $q$ as
\[
q>\frac{-(\frac{1}{4}\gamma^2+\frac{3}{4}\gamma)r+\sqrt{(\frac{1}{4}\gamma^2+\frac{3}{4}\gamma)^2r^2+\left(\epsilon-\left(\frac{1}{4}\gamma^2+\frac{1}{2}\gamma\right)r^2\right)\left(\frac{1}{2}\gamma+1\right)^2}}{(\frac{1}{2}\gamma+1)^2},  
\]
if $\epsilon\geq \left(\frac{1}{4}\gamma^2+\frac{1}{2}\gamma\right)r^2$ or $q\geq 0$, otherwise. These constraints will ensure the stability condition $\psi(q,r)<\gamma-\frac{2}{N}$ and the permanent survival of the inter-bank system. 

\section{Conclusion}\label{conclusion}
We study the linear quadratic regulator subject to the coupled CIR type diffusions. Given some regularity conditions, the admissible Markov Nash equilibrium creates system instability since the deposit rate given by the equilibrium diminishes the individual growth rate. A regulator must govern the growth rate of the ensemble total in order to prevent system crash. This problem can be extended in some directions. First, the behavior of simple lending and borrowing can be replaced with the system under relative performance concern proposed by \cite{Touzi2015}. Second, it would be interesting to study the delay obligation based on the system referring to \cite{Carmona-Fouque2016}. Due to the character of CIR type processes, we must tackle the admissibility of the equilibrium. Third, an individual bank may have a more restricted admissible set satisfying 
\[
X^{j,\alpha^i}_t\geq 0,\; \forall j\neq i\; {\mathrm{and}}\;X^{i,\alpha^i}_t> 0\;a.s..
\]  
The explicit solution of the admissible equilibrium would be of great interest.  The above extensions are ongoing research topics.   

\section*{Acknowledgement}
The author would like to thank  Jean-Pierre Fouque and Shuenn-Jyi Sheu for the conversations and suggestion on this subject.

{\appendix
\section{Verification Theorem}\label{Appex-1}
{\theorem\label{Ver-Thm}(Verification Theorem)\\
Given the optimal strategies $\hat\alpha^j$ given by (\ref{optimal-finite-ansatz}) for all $j\neq i$, $V^i$ given by (\ref{ansatz-Vi}) is the value function associated to the problem (\ref{value-function}) subject to (\ref{coupled}) and $\hat\alpha^i$  is the optimal strategy for bank $i$ and also the Markov Nash equilibrium. 
} 
\begin{proof}
Denote by $\tilde\alpha$ an admissible strategy obtained from the strategy $\hat\alpha=(\hat\alpha^1,\cdots,\hat\alpha^N)$ by replacing the $i$-th component by $\tilde\alpha^i$. That is, 
$$
\tilde\alpha=(\hat\alpha^1,\cdots,\hat\alpha^{i-1}, \alpha^i,\hat\alpha^{i+1},\cdots\hat\alpha^N).
$$
When $\alpha_t$ is replaced by $\tilde \alpha_t$ in (\ref{coupled}), we denote the solution  by 
$$
\tilde{X}_t=(\tilde{X}_t^1, \tilde{X}_t^2, \cdots, \tilde{X}_t^N).
$$
When $\alpha_t$ is replaced by $\hat \alpha_t$ in (\ref{coupled}), we denote the solution  by 
$$
\hat{X}_t=(\hat{X}_t^1, \hat{X}_t^2, \cdots, \hat{X}_t^N).
$$
We want to prove the following. For any admissible strategy $\tilde\alpha$, we have
\be \label{upper}
V^i(t,x)\leq \EE_{t,x}\left\{\int_t^T f_i(\tilde{X}_t, \alpha^i_s)ds+g_i(\tilde{X}_T)\right\}.
\en
and for $\hat{\alpha}$, we have
\be \label{optim}
V^i(t,x)= \EE_{t,x}\left\{\int_t^T f_i(\hat{X}_t, \hat{\alpha}^i_s)ds+g_i(\hat{X}_T)\right\},
\en
Using (\ref{upper}) and (\ref{optim}) implies that $\hat{\alpha}^i$ is an optimal strategy for $i$-th bank  when other banks use $\hat{\alpha}^j, j\neq i$ as their strategies.
We first consider (\ref{upper}) and (\ref{coupled})  for $\alpha_t=\tilde{\alpha}_t$. We can assume that 
\be \label{condition}
\EE_{t,x}\left\{\int_t^T f_i(\tilde{X}_t, \alpha^i_s)ds\right\}<\infty,
\en
otherwise, (\ref{upper}) holds  automatically. Define the exit time 
\[
\theta_M=\inf\{t;\;  |\tilde{X}_t|\geq M \}.
\]
%\[
%\theta=\inf\{t;\; |\tilde{X}_t|<\frac{1}{N}\; {\mathrm{or}}\; |X_t|> N \}.
%\]
 Under  (\ref{condition}), we shall prove the following properties:
\be \label{nonexplosion}
P(\theta_M\leq T)\rightarrow 0,\ M\rightarrow \infty.
\en
and
\be \label{ui}
\EE_{t, x}[\max_{t\leq s\leq T}|\tilde{X}_s|^2]<\infty.
\en
We postpone the proof of (\ref{nonexplosion}) and (\ref{ui}). We now prove (\ref{upper}). 

Given  $\hat\alpha^j$  by (\ref{optimal-finite-ansatz}) for all $j\neq i$, applying It\'o's formula on $V^i(s, \tilde{X}_s)$,  we obtain 

\ba
\nonumber&& dV^i(s,\tilde{X}_{s})\\
\nonumber&=&\bigg\{\partial_sV^i(s,\tilde{X}_s)+\sum_{j\neq i} \bigg(a(\overline{\tilde{X}}_s-\tilde{X}^j_s)+\gamma_s+{\hat\alpha^j_s(s, \tilde{X}_s) }\bigg)\partial_{x^j}V^i(s,\tilde{X}_s)\\
\nonumber&&+\bigg(a(\overline{\tilde{X}}_s-\tilde{X}^i_s)+\gamma_s+{\tilde\alpha^i}_s\bigg)\partial_{x^i}V^i(s,\tilde{X}_s)+ 2\sum_{j=1}^N   \tilde{X}^j_s \partial_{x^jx^j}V^i(s, \tilde{X}_s)\bigg\}ds \\
\nonumber&&+\sum_{j=1}^N2\sqrt{\tilde{X}^j_s}\partial_{x^j}V^i(s, \tilde{X}_s) dW^j_s,
\ea
and then
\ba
\nonumber&&V^i(T\wedge\theta_M,\tilde{X}_{T\wedge\theta_M})\\
\nonumber&=&V^i(t,x)+\int_t^{T\wedge\theta_M}\bigg\{\partial_sV^i(s,\tilde{X}_s)+\sum_{j\neq i} \bigg(a(\overline{\tilde{X}}_s-\tilde{X}^j_s)+\gamma_s+{\hat\alpha^j_s(s, \tilde{X}_s) }\bigg)\partial_{x^j}V^i(s,\tilde{X}_s)\\
\nonumber&&+\bigg(a(\overline{\tilde{X}}_s-\tilde{X}^i_s)+\gamma_s+{\tilde\alpha^i}_s\bigg)\partial_{x^i}V^i(s,\tilde{X}_s)+ 2\sum_{j=1}^N   \tilde{X}^j_s \partial_{x^jx^j}V^i(s, \tilde{X}_s)\bigg\} ds\\
\nonumber&&+\int_t^{T\wedge\theta_M}\sum_{j=1}^N2\sqrt{\tilde{X}^j_s}\partial_{x^j}V^i(s, \tilde{X}_s) dW^j_s.
\ea
Taking the expectation on both sides and using 
\ba \label{positive}
\nonumber && \partial_{t}V^i+\sum_{j\neq i}\bigg(a(\overline{x}-x^j)+\gamma_t+{\hat\alpha^j }\bigg)\partial_{x^j}V^i+\bigg(a(\overline{x}-x^i)+\gamma_t+{\alpha^i}\bigg)\partial_{x^i}V^i\\
&&+ 2\sum_{j=1}^N   x^j \partial_{x^jx^j}V^i+\frac{(\alpha^{i})^2}{2}-q\alpha^{i}\left(\overline{x}-x^i\right)+\frac{\epsilon}{2}(\overline x-x^i)^2 \geq 0
\ea
give
\ba
\nonumber  V^i(t,x) &\leq&\EE\bigg\{\int_t^{T\wedge\theta_M}\left(\frac{(\alpha^{i}_s)^2}{2}-q\alpha^{i}_s\left(\overline{\tilde X}_s-\tilde X^i_s\right)+\frac{\epsilon}{2}(\overline {\tilde X}_s-\tilde X^i_s)^2\right)ds \\
&&+ V^i(T\wedge\theta_M,\tilde X_{T\wedge\theta_M})\bigg\}.\label{ver-ineq}.
\ea
(\ref{ui}) implies  $V^i(T\wedge\theta_M,X_{T\wedge\theta_M})$ is uniformly integrable, since
$$
|V^i(T\wedge\theta_M,\tilde X_{T\wedge\theta_M})|\leq c(1+ \sup_{t\leq s\leq T} |\tilde{X}_s|^2).
$$ 
Together with (\ref{nonexplosion}), we obtain as $M\rightarrow\infty$
$$
\EE_{t, x}[ V^i(T\wedge\theta_M,\tilde X_{T\wedge\theta_M})]\rightarrow \EE_{t, x}[g(\tilde{X}_T)].
$$
We note
\ban
\nonumber && \EE\bigg\{\int_t^{T\wedge\theta_M}\left(\frac{(\alpha^{i}_s)^2}{2}-q\alpha^{i}_s\left(\overline{\tilde{X}}_s-\tilde{X}^i_s\right)+\frac{\epsilon}{2}(\overline{ \tilde{X}}_s-\tilde{X}^i_s)^2\right)ds \bigg\}\\
&&  \leq \EE\bigg\{\int_t^{T}\left(\frac{(\alpha^{i}_s)^2}{2}-q\alpha^{i}_s\left(\overline{\tilde{X}}_s-\tilde{X}^i_s\right)+\frac{\epsilon}{2}(\overline{ \tilde{X}}_s-\tilde{X}^i_s)^2\right)ds \bigg\},
\ean
since the integrand is nonnegative. 
Putting the above results together, we can obtain
\ba
V^i(t,x_t) \leq\EE\bigg\{\int_t^{T}\left(\frac{(\alpha^{i}_s)^2}{2}-q\alpha^{i}_s\left(\overline{\tilde X}_s-\tilde X^i_s\right)+\frac{\epsilon}{2}(\overline {\tilde X}_s-\tilde X^i_s)^2\right)ds 
+g(\tilde{X}_T))\bigg\}.\label{V(t,x)}
\ea
This completes the proof of (\ref{upper}).

We now prove (\ref{nonexplosion}).  Recall
$$
\tilde{Y}_t=\tilde{X}_t^1+\tilde{X}_t^2++\cdots+ \tilde{X}_t^N =N\overline{\tilde{X}}_t.
$$
%\[
%\theta^{\tilde{Y}}_{\widetilde M}=\inf\{t;\;  \tilde{Y}_t \notin [0,\widetilde M] \}.
%\]
%where 
The monetary reserve $\tilde{X}^i_t$ is given 
\be\label{ver-Xit}
  d\tilde{X}^i_t=\left(a(\frac{1}{N}\tilde{Y}_t-\tilde{X}^i_t)+\gamma_t+\alpha^i_t\right)dt+2\sqrt{\tilde{X}^i_t}dW^i_t,
\en
and given $\hat\alpha_t^j$ for $j\neq i$, the total monetary reserve $\tilde{Y}_t$ is governed by  
\[
 d\tilde{Y}_t =(N\gamma_t+\sum_{j\neq i}\hat\alpha_t^j+\alpha^i_t)dt++2\sqrt{  \tilde{Y}_t}d\widetilde W_t,
\]
where
$$
d \widetilde W_t=\frac{1}{\sqrt{  \tilde{Y}_t}}\sum_{j=1}^N \sqrt{\tilde{X}^j_t} dW_t^j
$$
is a Brownian motion and  
$$
d \tilde{X}^i_t \cdot d\tilde{Y}_t= 4 \tilde{X}^i_t dt.
$$
Using (28), we have
\ba
\nonumber \sum_{j\neq i} \hat{\alpha}^j  &=& -(N-1) \psi_t^c-(q+(1-\frac 1N) \eta_t^c)(\overline{\tilde{X}}_t-\tilde{X}^i_t)\\
\nonumber                                            &= & -(N-1) \psi_t^c-(q+(1-\frac 1N) \eta_t^c)(\frac 1N\tilde{Y}_t-\tilde{X}^i_t).
\ea
The equation for  $\tilde{Y}_t$  becomes
\[
 d\tilde{Y}_t =(N\gamma_t-(N-1)\psi_t ^c-(q+(1-\frac 1N) \eta_t^c)(\frac 1N\tilde{Y}_t-\tilde{X}^i_t)+\alpha^i_t)dt++2\sqrt{  \tilde{Y}_t}d\widetilde W_t.
\]

 Note that  $\tilde{X}^i_t\geq 0$ leads to $\tilde{Y}_t\geq \tilde{X}^j_t$ for $j=1,\cdots,N$ and $0\leq t \leq T$, since $\alpha^i$ is admissible implies $\tilde{X}^j_t\geq 0$ for all $j$ and for all  $0\leq t\leq T$ .  Applying It\^o fomula to $\tilde{Y}_t^2+ (\tilde{X}^i_t)^2$, we have
\ba
\nonumber d( \tilde{Y}_t^2+ (\tilde{X}^i_t)^2) &=& ( a(t) (\tilde{Y}_t)^2+ b(t) \tilde{Y}_t\tilde{X}^i_t + c(t) (\tilde{X}^i_t)^2 + (d(t)+ 2\alpha_t^i) \tilde{Y}_t+( e(t)+2\alpha_t^i )\tilde{X}^i_t) dt  \\
                                                                          &   & + 4\tilde{Y}_t\sqrt{\tilde{Y}_t} d\tilde{W}_t+ 4\tilde{X}_t^i\sqrt{\tilde{X}_t^i} dW^{i}_t,\label{Y+X}
\ea
where
\ba
\nonumber & & a(t)=-2\frac 1N (q+ (1-\frac 1N)\eta_t^c), \ b(t)= 2 (q+ (1-\frac 1N)\eta_t^c)+ 2a \frac 1N,\\
\nonumber &  & c(t)=-2a,\ d(t)=2(N\gamma_t-(N-1)\psi_t^c+ 2),\ e(t)=(2 \gamma_t +4).
\ea 
We take $\beta>0$ large enough such that 
\ba
\nonumber \beta\geq 2\sup_{0\leq t\leq T} \{ |a(t)|+ |b(t)| + |c(t)| + |d(t)|+|e(t)| +1\},
\ea
and then consider $e^{-\beta t} (\tilde{Y}_t^2+ (\tilde{X}^i_t)^2)$ written as
\ba
\nonumber e^{-\beta t} (\tilde{Y}_t^2+ (\tilde{X}^i_t)^2)&=&\{ -\beta( \tilde{Y}_t^2+ (\tilde{X}^i_t)^2)) + a(t) (\tilde{Y}_t)^2+ b(t) \tilde{Y}_t\tilde{X}^i_t + c(t) (\tilde{X}^i_t)^2 + (d(t)+ 2\alpha_t^i) \tilde{Y}_t\\
\nonumber                                                                       & & + ( e(t) +  2\alpha_t^i)) \tilde{X}^i_t\} e^{-\beta t} dt  + e^{-\beta t} (4\tilde{Y}_t\sqrt{\tilde{Y}_t} d\tilde{W}_t+ 4\tilde{X}_t^i\sqrt{\tilde{X}_t^i} dW^{i}_t)\\
\nonumber &\leq & \{ -\frac 12 \beta( \tilde{Y}_t^2+ (\tilde{X}^i_t)^2-1) + |\alpha_t^i|^2 \} e^{-\beta_t}dt+   e^{-\beta t} (4\tilde{Y}_t\sqrt{\tilde{Y}_t} d\tilde{W}_t+ 4\tilde{X}_t^i\sqrt{\tilde{X}_t^i} dW^{i}_t),
\ea
leading to
\ba
\nonumber& & e^{-\beta(t\wedge\theta_M)}(\tilde{Y}^2_{t\wedge\theta_M}+( \tilde{X}^i)^2_{t\wedge\theta_M})\\
\nonumber&\leq &(\tilde{Y}_0^2+ (\tilde{X}_0^i)^2) + \frac{\beta}{2} t + \int_0^{t\wedge\theta_M} e^{-\beta s} |\alpha_s^i|^2 ds -  \frac {\beta}{2} \int_0^{t\wedge\theta_M} e^{-\beta s}(\tilde{Y}_s^2+ (\tilde{X}^i_s)^2) ds\\
\nonumber & & +  \int_0^{t\wedge\theta_M}e^{-\beta t} (4\tilde{Y}_t\sqrt{\tilde{Y}_t} d\tilde{W}_t+ 4\tilde{X}_t^i\sqrt{\tilde{X}_t^i} dW^{i}_t).
\ea
Taking expectation on both sides gives 
\be
 e^{-\beta t} M^2\PP(\theta_M\leq t)\leq \tilde{Y}_0^2+(\tilde{X}^i_0)^2+ \frac{\beta}{2} t+\EE[\int_0^t |\alpha^i_s|^2 ds] -\frac{\beta}{2}\EE[ \int_0^{t\wedge\theta_M} e^{-\beta s}(\tilde{Y}_s^2+ (\tilde{X}^i_s)^2) ds ].
\en
Take $t=T$, as $  M\rightarrow \infty$, we have (\ref{nonexplosion}). The above relation also implies
\ba
\nonumber \frac{\beta}{2}\EE[ \int_0^{T\wedge\theta_M} e^{-\beta s}(\tilde{Y}_s^2+ (\tilde{X}^i_s)^2) ds ]\leq  \tilde{Y}_0^2+(\tilde{X}^i_0)^2+ \frac{\beta}{2} T + \ \EE[\int_0^T |\alpha^i_s|^2 ds] 
\ea
This is true for any $M$. By letting $M\rightarrow \infty$ and using (\ref{nonexplosion}), we obtain
\ba\label{inequality-last-step}
  \frac{\beta}{2}\EE[ \int_0^{T} e^{-\beta s}(\tilde{Y}_s^2+ (\tilde{X}^i_s)^2) ds ]\leq  \tilde{Y}_0^2+(\tilde{X}^i_0)^2+ \frac{\beta}{2} T +  \EE[\int_0^T |\alpha^i_s|^2 ds] .
\ea
We are ready to prove (\ref{ui}).  Applying Doob's martingale inequality and Cauchy-Schwarz inequality to (\ref{ver-Xit}), we get 
\ba
\nonumber\EE\sup_{t\leq s\leq T} |\tilde{X}^i_s|^2&\leq& 2\EE\left[\int_t^{T}\left(a(\frac{1}{N}|\tilde{Y}_s|+|\tilde{X}^i_s|)+\gamma_s\right)ds\right]^2\\
\nonumber&&+2\EE\left[\int_t^{T}|\alpha^i_s|ds\right]^2
+2\EE\left[\sup_{t\leq s\leq T}\left(\int_t^{T}2\sqrt{\tilde{X}^i_s}ds\right)^2\right]\\
\nonumber &\leq& C_1T\EE\int_t^{T}(|\tilde{Y}_s|^2+|\tilde{X}_s^i|^2+|\gamma_s|^2+|\alpha_s^i|^2)ds+C_2
\EE\left[\int_t^{T}\tilde{X}^i_sds\right]< \infty
\ea
where $C_1$ and $C_2$ are two positive constants. Here in the last step, we use (\ref{inequality-last-step}). This proves (\ref{ui}).

In the following, we want to show $\hat{\alpha}^i$ is an admissible strategy for the $i$-th bank and its value is given by $V^i(t, x)$. Hence  it is an optimal strategy.
In this case,  given $\hat\alpha^j$ for $j\neq i$, the monetary reserve for $i$-th bank is $\hat{X}^i_t$ which is governed by the equation
\ba
\nonumber d\hat{X}^i_t&=&\left((a+q+(1-\frac 1N)\eta_t^c)(\frac{1}{N}\hat{Y}_t-\hat{X}^i_t)+\gamma_t-\psi_t^c\right)dt+2\sqrt{\hat{X}^i_t}dW^i_t\\
&=&\left(\frac{a+q+(1-\frac 1N)\eta_t^c}{N}\sum_{j=1}^N(\hat X^j_t-\hat X^i_t)+\gamma_t-\psi_t^c\right)dt+2\sqrt{\hat{X}^i_t}dW^i_t,  
\ea
the total monetary reserve $\tilde{Y}_t$ is written as 
\[
 d\hat{Y}_t =N(\gamma_t-\psi_t^c)dt++2\sqrt{  \hat{Y}_t}d\widetilde W_t,
\]
where
$$
d \widetilde W_t=\frac{1}{\sqrt{  \hat{Y}_t}}\sum_{j=1}^N \sqrt{\hat{X}^j_t} dW_t^j
$$
is a Brownian motion. Under the conditions (\ref{regularity-cond-1}) and (\ref{regularity-cond-2}), according to the Proposition 2.1 in \cite{Fouque-Ichiba}, we have $\hat X^i_t\geq 0$ almost surely for $i=1,\cdots,N$. This implies the admissibility of the strategy $\hat\alpha$ and the regularity of the total monetary reserve $\hat{Y}_t$. Using the above argument ( for the estimates of  $\tilde{Y}_t, \tilde{X}^i_t$), we can show, for a suitable $\beta>0$,
$$
E[\int_0^T e^{-\beta s} \hat{Y}_s^2 ds]<\infty,
$$
by considering the Ito's formula for $\hat{Y}_s^2$ (see (\ref{Y+X})). Then we can also show,
$$
E[\sup_{0\leq  s\leq T} |\hat{X}^i_s|^2]<\infty.
$$
Then as (\ref{V(t,x)}), we can derive
$$
V^i(t, x)=E[\int_t^T (\frac{(\alpha_s^i)^2}{2} - q \hat{\alpha}^i_s (\bar{\hat{X}}_s- \hat{X}^i_s) + \frac{\epsilon}{2}  (\bar{\hat{X}}_s- \hat{X}^i_s)^2 ) ds + g(\hat{X}_T)].
$$
That is, $\hat{\alpha}^i$ is an optimal strategy for $i$-th bank when $\hat{\alpha}^j, j\neq i$ are used for $j$-th bank.
The proof is complete.

\end{proof}

{\theorem\label{Ver-Thm-inf}  
Given the optimal strategies $\hat\alpha^j$ written as (\ref{optimal-control-inf}) for all $j\neq i$, $V^i$ and $\hat\alpha^i$ given by (\ref{ansatz-inf})  and (\ref{optimal-control-inf}) respectively are the value function and the corresponding optimal strategy which is also the Markov Nash equilibrium  associated to the problem (\ref{value-inf}-\ref{coupled-inf}).  }
%suppose that $\nu\in C^{2}(\RR)$ is nonnegative satisfying the growth condition in Chapter III.9 of \cite{Fleming2006}
%\be\label{growth-inf}
%|\nu(x)|\leq K(1+|x|^m),
%\en
%for some constants $K$ and $m$ and solves the HJB equation 
%\ba\label{HJB-inf-1}
%\nonumber -r\nu&+&\inf_{\alpha^i}\bigg\{
%\sum_{j\neq i}\bigg(a(\overline{x}-x^j)+\gamma+{\hat\alpha^j(t,x)}\bigg)\partial_{x^j}V^i
%+ \bigg(a(\overline{x}-x^i)+\gamma+{\alpha^i }\bigg)\partial_{x^i}V^i\\
%&+&2\sum_{j=1}^N   x^j \partial_{x^jx^j}V^i+\frac{(\alpha^{i})^2}{2}-q\alpha^{i}\left(\overline{x}-x^i\right)+\frac{\epsilon}{2}(\o x-x^i)^2\bigg\} =0.
%\ea
%Then $\nu=V^i$ and 
%\ba
%\nonumber \hat\alpha^i &:=&\mathrm{argmin} \bigg\{
%\sum_{j\neq i}\bigg(a(\overline{x}-x^j)+\gamma+{\hat\alpha^j(t,x)}\bigg)\partial_{x^j}V^i
%+ \bigg(a(\overline{x}-x^i)+\gamma+{\alpha^i}\bigg)\partial_{x^i}V^i\\
%&+&2\sum_{j=1}^N   x^j \partial_{x^jx^j}V^i+\frac{(\alpha^{i})^2}{2}-q\alpha^{i}\left(\overline{x}-x^i\right)+\frac{\epsilon}{2}(\o x-x^i)^2\bigg\} \label{optimal alpha^i-inf},
%\ea
%is an optimal strategy.}
\begin{proof}
Adopting the notations in Theorem \ref{Ver-Thm}, denote $\tilde\alpha$ as an admissible strategy written as
$$
\tilde\alpha=(\hat\alpha^1,\cdots,\hat\alpha^{i-1},\alpha^i,\hat\alpha^{i+1},\cdots\hat\alpha^N).
$$ 
Similarly, when $\alpha_t$ is replaced by $\tilde \alpha_t$ in (\ref{coupled}), the solution is written as 
$$
\tilde{X}_t=(\tilde{X}_t^1, \tilde{X}_t^2, \cdots, \tilde{X}_t^N),
$$
and when $\alpha_t$ is replaced by $\hat \alpha_t$ in (\ref{coupled}), the solution is given by 
$$
\hat{X}_t=(\hat{X}_t^1, \hat{X}_t^2, \cdots, \hat{X}_t^N).
$$
We again need to verify that for any admissible strategy $\tilde\alpha$, the solution $V^i$ satisfies 
\be \label{upper-inf}
V^i(t,x)\leq \EE_{t,x}\left\{\int_0^\infty e^{-rt}f_i(\tilde{X}_t, \alpha^i_s)ds \right\},
\en
and for $\hat{\alpha}$.
\be \label{optim-inf}
V^i(t,x)= \EE_{t,x}\left\{\int_0^\infty e^{-rt}f_i(\hat{X}_t, \hat{\alpha}^i_s)ds \right\},
\en
The equations (\ref{upper-inf}) and (\ref{optim-inf}) lead to $\hat{\alpha}^i$ is an optimal strategy for $i$-th bank  given $\hat{\alpha}^j$ for $j\neq i$ as their strategies. For the proof of (\ref{upper-inf}), we can assume that
\be \label{condition-inf}
\EE_{t,x}\left\{\int_0^\infty e^{-rt}f_i(\tilde{X}_t, \alpha^i_s)ds\right\}<\infty,
\en
otherwise, (\ref{upper}) holds  automatically. Define the exit time 
\[
\theta_M=\inf\{t;\; |\tilde X_t| \geq M \},
\]
or $\theta_M=+\infty$ if $X_t$ never exit for all $t\geq 0$. Given the optimal strategies $\hat\alpha^j$ written as (\ref{optimal-control-inf}) for all $j\neq i$, using It\'o's formula on $e^{-rt}V^i$, we obtain   
\ba
\nonumber&&e^{-r{T\wedge\theta_M}}V^i(\tilde X_{T\wedge\theta_M})\\
\nonumber &=&V^i(x)+\int_0^{T\wedge\theta_M}e^{-rs}\bigg\{-rV^i(\tilde X_s)+\sum_{j\neq i} \bigg(a(\overline{\tilde X}_s-\tilde X^j_s)+\gamma+{\hat\alpha^j_s }\bigg)\partial_{x^j}V^i(\tilde X_s)\nonumber\\
\nonumber&&+\bigg(a(\overline{\tilde X}_s-\tilde X^i_s)+\gamma+\alpha^i_s\bigg)\partial_{x^i}V^i(\tilde X_s)+ 2\sum_{j=1}^N   \tilde X^j_s \partial_{x^jx^j}V^i(\tilde X_s)\bigg\} ds\\
&&+\int_t^{T\wedge\theta_M}e^{-rs}\sum_{i=1}^N2\sqrt{\tilde X^j_s}\partial_{x^j}V^i(\tilde X_s) dW^j_s.
\ea
Taking the expectation of both sides and using  
\ba
\nonumber &&  -rV^i+
\sum_{j\neq i}\bigg(a(\overline{x}-x^j)+\gamma+{\hat\alpha^j }\bigg)\partial_{x^j}V^i+ \bigg(a(\overline{x}-x^i)+\gamma+{\alpha^i}\bigg)\partial_{x^i}V^i\\
&&+2\sum_{j=1}^N   x^j \partial_{x^jx^j}V^i+\frac{(\alpha^{i})^2}{2}-q\alpha^{i}\left(\overline{x}-x^i\right)+\frac{\epsilon}{2}(\overline x-x^i)^2\geq 0
\ea
%\ba
%\nonumber \hat\alpha^i &:=&\mathrm{argmin} \bigg\{-rV^i+
%\sum_{j\neq i}\bigg(a(\overline{x}-x^j)+\gamma+{\hat\alpha^j }\bigg)\partial_{x^j}V^i+ \bigg(a(\overline{x}-x^i)+\gamma+{\alpha^i}\bigg)\partial_{x^i}V^i\\
%&&+2\sum_{j=1}^N   x^j \partial_{x^jx^j}V^i+\frac{(\alpha^{i})^2}{2}-q\alpha^{i}\left(\overline{x}-x^i\right)+\frac{\epsilon}{2}(\o x-x^i)^2\bigg\} \label{optimal alpha^i-inf},
%\ea
give 
\ban
V^i(x)&\leq&\EE\bigg\{\int_0^{T\wedge\theta_M}e^{-rs}\left(\frac{(\alpha^{i}_s)^2}{2}-q\alpha^{i}_s\left(\overline{\tilde X}_s-\tilde X^i_s\right)+\frac{\epsilon}{2}(\overline {\tilde X}_s-\tilde X^i_s)^2\right)ds\\
&&+e^{-r{T\wedge\theta_M}}V^i(\tilde X_{T\wedge\theta_M})\bigg\}.
\ean
Using the ansatz $V^i$ given by (\ref{ansatz-inf}) gives 
\[
V^i(\tilde X_{T\wedge\theta_M})\leq K(1+\sup_{0\leq t\leq T}|\tilde X_t|^2).
\]
Hence, according to the results in Theorem \ref{Ver-Thm},  when taking the limit as $M\rightarrow\infty$, we get the uniform integrability of $V^i$. 
\be\label{ineq-ver-inf}
V^i(x)\leq\EE\bigg\{\int_0^{T}e^{-rs}\left(\frac{(\alpha^{i}_s)^2}{2}-q\alpha^{i}_s\left(\overline{\tilde X}_s-\tilde X^i_s\right)+\frac{\epsilon}{2}(\overline {\tilde X}_s-\tilde X^i_s)^2\right)ds+e^{-r{T }}V^i(\tilde X_{T})\bigg\}.
\en
Referring to Chapter III.9 in \cite{Fleming2006},  using $\phi+L>0$, we obtain $V^i\geq 0$ satisfying  
\be\label{suff-cond-inf-1}
{\lim\sup}_{T\rightarrow\infty}e^{-rT}\EE V^i(\hat X_T)\geq0,
\en
leading to 
\[
V(x)\geq \EE_{x}\left\{\int_0^\infty e^{-rt}f_i(\hat X_s,   \alpha^i_s)ds\right\}.
\]
We now need to verify 
\be\label{suff-cond-inf-2}
{\lim\inf}_{T\rightarrow\infty}e^{-rT}\EE V^i(\tilde X_T)=0,
\en
implying
\[
V^i(x)\leq  \EE_{x}\left\{\int_0^\infty e^{-rt}f_i(\tilde X_s,   \alpha^i_s)ds\right\}.
\]
Recalling
$$
V^{i}(\tilde{X}_s)=\frac{\eta}{2}(\overline{\tilde{X}}_s-\tilde{X}^i_s)^2+ L(\overline{\tilde{X}}_s-\tilde{X}^i_s) + \phi \overline{\tilde{X}}_s +\mu,
$$
since 
$$
f_i(\tilde{X}_s, \alpha_s^i)\geq (\epsilon-q^2)(\overline{\tilde{X}}_s-\tilde{X}^i_s)^2,
$$
$$
f_i(\tilde{X}_s, \alpha_s^i)\geq \frac 12 (1-\frac{q^2}{\epsilon})\alpha_s^2,
$$
and 
$$
E[\int_0^{\infty} e^{-rs} f_i(\tilde{X}_s, \alpha_s^i) ds]<\infty,
$$
we have
$$
E[\int_0^T e^{-rs}  (\overline{\tilde{X}}_s-\tilde{X}^i_s)^2 ds]<\infty,
$$
and
$$
E[\int_0^T e^{-rs}  (\alpha_s)^2 ds]<\infty,
$$
implying
$$
\liminf _{T\rightarrow \infty} e^{-rT} E[ (\overline{\tilde{X}}_T-\tilde{X}^i_T)^2]=0.
$$
We need to show 
$$
\lim_{T\rightarrow \infty} e^{-rT} E[\overline{\tilde{X}}_T]=0.
$$
Since
$$
d\tilde{Y}_t=(N\gamma-(N-1) \psi- (q+(1-\frac1N)\eta) (\overline{\tilde{X}}_t-\tilde{X}^i_t)+ \alpha_t^i) dt + 2\sqrt{\tilde{Y}_t} d\tilde{W}_t,
$$
where $\tilde{W}$ is a Brownian motion (see the argument in Theorem 1), we have 
$$
E[\tilde{Y}_T]=E[\tilde{Y}_0]+ (N\gamma-(N-1) \psi)T -  (q+(1-\frac1N)\eta)\int_0^T E[(\overline{\tilde{X}}_s-\tilde{X}^i_s) ds]+ \int_0^T  E[\alpha_s^i]ds.
$$
Then 
$$
\begin{array}{l}
e^{-rT}E[ \int_{0}^T  |\overline{\tilde{X}}_s-\tilde{X}^i_s| ds]\\
=e^{-rT}E[ \int_{0}^T e^{\frac r2 s} e^{-\frac r2 s} |\overline{\tilde{X}}_s-\tilde{X}^i_s| ds]\\
\leq e^{-rT}(\int_{0}^Te^{rs} ds)^{1/2}( E[\int_{0}^T e^{-rs} |\overline{\tilde{X}}_s-\tilde{X}^i_s|^2 ds])^{1/2}\\
\leq \frac{1}{\sqrt{r}} e^{-\frac 12 rT}(E[ \int_{0}^T e^{-rs} |\overline{\tilde{X}}_s-\tilde{X}^i_s|^2 ds])^{1/2}\\
\rightarrow 0,\ T\rightarrow \infty.
\end{array}
$$
Similarly,
$$
e^{-rT}E[ \int_{0}^T | \alpha^i_s| ds] \rightarrow 0,\ T\rightarrow \infty.
$$
Since $\tilde{Y}_T=N \overline{\tilde{X}}_T$, putting the above results together, we can obtain
$$
\lim_{T\rightarrow \infty} e^{-rT} E[\overline{\tilde{X}}_T]=0.
$$ 
Hence, the above results gives (\ref{suff-cond-inf-2}). 

Finally, referring to Proposition 2.1 \cite{Fouque-Ichiba}, the conditions (\ref{regularity-cond-1-inf}) and (\ref{regularity-cond-2-inf}) give the regularity condition of the total monetary reserve (\ref{Y-inf}) and the admissibility of the optimal strategies $\hat\alpha^i$ for $i=1,\cdots,N$. This proves that $V^i$ is the value function and $\hat\alpha^i$ is an optimal strategy. 
\end{proof}

}

\bibliographystyle{plainnat}
%\bibliographystyle{plain}
%\nocite{*}
\bibliography{references-delay-games-feller}

\end{document}